\documentclass[11pt]{article}
\usepackage[margin=1in,left=1in]{geometry}
\usepackage{amsmath}
\usepackage{amsfonts}
\usepackage{amssymb}
\usepackage{amsthm}
\usepackage{graphicx}
\usepackage{hyperref}
\usepackage[all]{xy}

\makeatletter
\newdimen\itex@wd%
\newdimen\itex@dp%
\newdimen\itex@thd%
\def\itexspace#1#2#3{\itex@wd=#3em%
\itex@wd=0.1\itex@wd%
\itex@dp=#2ex%
\itex@dp=0.1\itex@dp%
\itex@thd=#1ex%
\itex@thd=0.1\itex@thd%
\advance\itex@thd\the\itex@dp%
\makebox[\the\itex@wd]{\rule[-\the\itex@dp]{0cm}{\the\itex@thd}}}
\makeatother

\makeatletter
\newif\if@sup
\newtoks\@sups
\def\append@sup#1{\edef\act{\noexpand\@sups={\the\@sups #1}}\act}%
\def\reset@sup{\@supfalse\@sups={}}%
\def\mk@scripts#1#2{\if #2/ \if@sup ^{\the\@sups}\fi \else%
  \ifx #1_ \if@sup ^{\the\@sups}\reset@sup \fi {}_{#2}%
  \else \append@sup#2 \@suptrue \fi%
  \expandafter\mk@scripts\fi}
\def\tensor#1#2{\reset@sup#1\mk@scripts#2_/}
\def\multiscripts#1#2#3{\reset@sup{}\mk@scripts#1_/#2%
  \reset@sup\mk@scripts#3_/}
\makeatother

\makeatletter
\newbox\slashbox \setbox\slashbox=\hbox{$/$}
\def\itex@pslash#1{\setbox\@tempboxa=\hbox{$#1$}
  \@tempdima=0.5\wd\slashbox \advance\@tempdima 0.5\wd\@tempboxa
  \copy\slashbox \kern-\@tempdima \box\@tempboxa}
\def\slash{\protect\itex@pslash}
\makeatother

\def\clap#1{\hbox to 0pt{\hss#1\hss}}

\let\oldroot\root
\def\root#1#2{\oldroot #1 \of{#2}}
\renewcommand{\sqrt}[2][]{\oldroot #1 \of{#2}}

\DeclareSymbolFont{symbolsC}{U}{txsyc}{m}{n}
\SetSymbolFont{symbolsC}{bold}{U}{txsyc}{bx}{n}
\DeclareFontSubstitution{U}{txsyc}{m}{n}

\DeclareSymbolFont{stmry}{U}{stmry}{m}{n}
\SetSymbolFont{stmry}{bold}{U}{stmry}{b}{n}

\DeclareFontFamily{OMX}{MnSymbolE}{}
\DeclareSymbolFont{mnomx}{OMX}{MnSymbolE}{m}{n}
\SetSymbolFont{mnomx}{bold}{OMX}{MnSymbolE}{b}{n}
\DeclareFontShape{OMX}{MnSymbolE}{m}{n}{
    <-6>  MnSymbolE5
   <6-7>  MnSymbolE6
   <7-8>  MnSymbolE7
   <8-9>  MnSymbolE8
   <9-10> MnSymbolE9
  <10-12> MnSymbolE10
  <12->   MnSymbolE12}{}

\makeatletter
\def\re@DeclareMathSymbol#1#2#3#4{%
    \let#1=\undefined
    \DeclareMathSymbol{#1}{#2}{#3}{#4}}
\re@DeclareMathSymbol{\neArrow}{\mathrel}{symbolsC}{116}
\re@DeclareMathSymbol{\neArr}{\mathrel}{symbolsC}{116}
\re@DeclareMathSymbol{\seArrow}{\mathrel}{symbolsC}{117}
\re@DeclareMathSymbol{\seArr}{\mathrel}{symbolsC}{117}
\re@DeclareMathSymbol{\nwArrow}{\mathrel}{symbolsC}{118}
\re@DeclareMathSymbol{\nwArr}{\mathrel}{symbolsC}{118}
\re@DeclareMathSymbol{\swArrow}{\mathrel}{symbolsC}{119}
\re@DeclareMathSymbol{\swArr}{\mathrel}{symbolsC}{119}
\re@DeclareMathSymbol{\nequiv}{\mathrel}{symbolsC}{46}
\re@DeclareMathSymbol{\Perp}{\mathrel}{symbolsC}{121}
\re@DeclareMathSymbol{\Vbar}{\mathrel}{symbolsC}{121}
\re@DeclareMathSymbol{\sslash}{\mathrel}{stmry}{12}
\re@DeclareMathSymbol{\bigsqcap}{\mathop}{stmry}{"64}
\re@DeclareMathSymbol{\biginterleave}{\mathop}{stmry}{"6}
\re@DeclareMathSymbol{\invamp}{\mathrel}{symbolsC}{77}
\re@DeclareMathSymbol{\parr}{\mathrel}{symbolsC}{77}
\makeatother

\makeatletter
\def\Decl@Mn@Delim#1#2#3#4{%
  \if\relax\noexpand#1%
    \let#1\undefined
  \fi
  \DeclareMathDelimiter{#1}{#2}{#3}{#4}{#3}{#4}}
\def\Decl@Mn@Open#1#2#3{\Decl@Mn@Delim{#1}{\mathopen}{#2}{#3}}
\def\Decl@Mn@Close#1#2#3{\Decl@Mn@Delim{#1}{\mathclose}{#2}{#3}}
\Decl@Mn@Open{\llangle}{mnomx}{'164}
\Decl@Mn@Close{\rrangle}{mnomx}{'171}
\Decl@Mn@Open{\lmoustache}{mnomx}{'245}
\Decl@Mn@Close{\rmoustache}{mnomx}{'244}
\makeatother

\makeatletter
\DeclareRobustCommand\widecheck[1]{{\mathpalette\@widecheck{#1}}}
\def\@widecheck#1#2{%
    \setbox\z@\hbox{\m@th$#1#2$}%
    \setbox\tw@\hbox{\m@th$#1%
       \widehat{%
          \vrule\@width\z@\@height\ht\z@
          \vrule\@height\z@\@width\wd\z@}$}%
    \dp\tw@-\ht\z@
    \@tempdima\ht\z@ \advance\@tempdima2\ht\tw@ \divide\@tempdima\thr@@
    \setbox\tw@\hbox{%
       \raise\@tempdima\hbox{\scalebox{1}[-1]{\lower\@tempdima\box
\tw@}}}%
    {\ooalign{\box\tw@ \cr \box\z@}}}
\makeatother

\makeatletter
\def\udots{\mathinner{\mkern2mu\raise\p@\hbox{.}
\mkern2mu\raise4\p@\hbox{.}\mkern1mu
\raise7\p@\vbox{\kern7\p@\hbox{.}}\mkern1mu}}
\makeatother

\newcommand{\gt}{>}

\newcommand{\R}{\ensuremath{\mathbb R}}
\newcommand{\N}{\ensuremath{\mathbb N}}
\newcommand{\Z}{\ensuremath{\mathbb Z}}
\newcommand{\Q}{\ensuremath{\mathbb Q}}

\renewcommand{\(}{\begin{equation*}}
\renewcommand{\)}{\end{equation*}}
\newcommand{\bea}{\begin{eqnarray*}}
\newcommand{\eea}{\end{eqnarray*}}

\theoremstyle{italics}
\newtheorem{theorem}{Theorem}[section]
\newtheorem{lemma}[theorem]{Lemma}
\newtheorem{prop}[theorem]{Proposition}
\newtheorem{cor}[theorem]{Corollary}

\theoremstyle{definition}
\newtheorem{defn}[theorem]{Definition}
\newtheorem{example}[theorem]{Example}

\theoremstyle{remark}
\newtheorem{remark}[theorem]{Remark}
\newtheorem{note[theorem]}{Note}

\begin{document}

%

\title{
Rational sphere valued
supercocycles
 in M-theory \\
 and type IIA string theory}

 \author{Domenico Fiorenza\thanks{Dipartimento di Matematica, La Sapienza Universit \`a di Roma
Piazzale Aldo Moro 2, 00185 Rome, Italy}, \;
     Hisham Sati\thanks{University of Pittsburgh, Pittsburgh, PA 15260, USA, and New York University,
     Abu Dhabi, UAE}, \;
     Urs Schreiber\thanks{Mathematics Institute of the Academy, {\v Z}itna 25, 115 67 Praha 1, Czech Republic; on leave at MPI Bonn}
     }

\maketitle

\begin{abstract}
We show that supercocycles on super $L_\infty$-algebras capture, at the rational level,
the twisted cohomological charge structure of the fields of M-theory and of type IIA string theory.
We show that rational 4-sphere-valued supercocycles for M-branes in M-theory descend
to supercocycles in type IIA string theory with coefficients in the free loop space of the 4-sphere,
to yield the Ramond-Ramond fields in the rational image of twisted K-theory, with the twist given by the B-field.
In particular, we derive the M2/M5 $\leftrightarrow$ F1/Dp/NS5 correspondence via dimensional reduction
of sphere-valued super-$L_\infty$-cocycles.

\end{abstract}

 \tableofcontents

\section{Introduction}

(Super)cocycles play an important role in the study of the geometric and topological
structures associated with physical theories (see \cite{dAI} for an earlier survey).  In \cite{Bouq}
we discussed cocycles of super $L_\infty$-algebras (super Lie $n$-algebras for arbitrary $n$)
forming the \emph{brane bouquet} that gives the WZW terms of all the
Green-Schwarz sigma models for all the branes in string theory and M-theory. This includes those with gauge fields
on their worldvolume, the D-branes and the M5-brane, which were missing in the classical \emph{brane scan}.

\medskip
In \cite{FSS15} we had shown that this approach allows deriving the rational image of a
twisted cohomology theory that unifies the M2-brane charges and the M5-brane charges (this is recalled in section \ref{Sec M} below).
Rationally this cohomology theory turns out to be represented by the 4-sphere, hence is \emph{cohomotopy} in degree 4.
This is in higher analogy to the familiar statement that the unification of D$p$-brane charges with the F1-brane charge
ought to be in twisted K-cohomology theory.
(That the fields of M-theory should take values in the 4-sphere was first suggested in
\cite{S1, S2}.)

\medskip
In section \ref{Sec red} we show, at the rational level, that indeed the twisted
M2/M5 charges in degree-4 cohomotopy in 11 dimensions dimensionally reduce to the twisted K-theory
of the F1/D$p$/NS5-brane charges in 10 dimensions (for $p \in \{0,2,4\}$), where the
dimensionally reduced cohomology theory is represented
by the rationalization of the homotopy quotient $\mathcal{L} S^4 /\!/ S^1$ of the free loop space of the 4-sphere.
In particular this exhibits a purely $L_\infty$-theoretic derivation, at the rational level,
of twisted K-theory as the home of the brane charges in type II string theory. The lift of this twisted charge
structure to M-theory has been an open problem. This may be viewed as one confirmation at the rational level of the
proposal in \cite{S1, S2} on the description of M-theory via twisted generalized cohomology.

\medskip
In the existing literature, the cocycles for the WZW terms of the D$p$-branes are instead constructed
separately as independent cocycles on extended super-Minkowski spacetime (see \cite{FSS15} for
references and for the super $L_\infty$-algebraic formulation). In section \ref{Sec IIA} we show that
the same $L_\infty$-descent mechanism which unifies the M2- and M5-brane charges also applies to the
separate D$p$-brane cocycles, and they descend to again a single cocycle with coefficients in (the rational image of) the relevant
truncation of  twisted K-theory.


\medskip
The techniques that we use are from geometric homotopy theory \cite{dcct}, cast in computationally
powerful algebraic language. Lecture notes accompanying the discussion here
may be found in \cite{StructureTheory}.
 We consider super $L_\infty$-algebras as in \cite{SSS09, Bouq}. These are a generalizations of
super Lie algebras to super Lie $n$-algebras, for arbitrary $n$, where instead of just a super Lie bracket we have brackets of all arities with
the Lie bracket being the binary one. More precisely, our construction takes place in the homotopy category of
super $L_\infty$-algebras, so that a morphism from a super $L_\infty$-algebra $\mathfrak{g}$ to a super  $L_\infty$-algebra $\mathfrak{h}$ will actually be a span of morphism
\[
\mathfrak{g}\xleftarrow{\sim} \tilde{\mathfrak{g}}\to \mathfrak{h}
\]
where $\tilde{\mathfrak{g}}\xrightarrow{\sim} \mathfrak{g}$ is a quasi-isomorphism, i.e.,
an $L_\infty$-morphism inducing an isomorphism of graded vector spaces at the level of cohomology
from $H^\bullet(\tilde{\mathfrak{g}})$ to $H^\bullet(\mathfrak{g})$. Passing from $\mathfrak{g}$ to $\tilde{\mathfrak{g}}$ is an example of resolution. This concept has many incarnations, depending on
the context (homotopic, fibrant, cofibrant, projective, injective). For us, what is important
is that is is a concept of equivalence within a category
between the object at hand and another (or a combination of such)
 that generally behaves in a more utilizable way within the same category.

\medskip
Furthermore, we will make constant use of the duality between (finite type) super $L_\infty$-algebras
and
differential graded-commutative super-algebras, identifying a super $L_\infty$-algebra $\mathfrak{g}$ with its
Chevalley-Eilenberg algebra ${\rm CE}(\mathfrak{g})$ as in \cite{SSS09}.
These Chevalley-Eilenberg algebras of super $L_\infty$-algebras are what are called \emph{FDA}s
in the supergravity literature (going back to \cite{DF}). The point of identifying these
as dual to super $L_\infty$-algebras is to make manifest their higher gauge theoretic nature
and the relevant homotopy theory, which is crucial for the results we present here.
For instance, for $p\in \N$, the line $(p+2)$-algebra $b^{p+1}\R$, i.e., the chain complex with $\R$ in degree
$p+1$ and zeros everywhere else, corresponds to the
Chevalley-Eilenberg algebra
$$
{\rm CE}(b^{p+1}\R):=\big( \mathbb{R}[ g_{p+2}];~ dg_{p+2}=0\big)\;,
$$
where the generator $g_{p+2}$ has degree $p+2$.

\medskip
Notice that ${\rm CE}(b^{p+1}\R)$ is the minimal Sullivan model for the rational space $B^{p+2}\mathbb{R}$,
reflecting the fact that $b^{p+1}\R$ is the $L_\infty$-algebra corresponding to the $\infty$-group
$B^{p+1}\R \simeq \Omega B^{p+2} \R$.
In order to amplify this relation between $L_\infty$-algebras and rational homotopy theory, we also write
$\mathfrak{l}(X)$, or simply $\mathfrak{l}X$, for the $L_\infty$-algebra whose CE-algebra is a given Sullivan model of finite type for some rational space $X$ :
$$
  \mathfrak{l}(X) \;=\;
  \mbox{$L_\infty$-algebra dual to given Sullivan model $(A_X, d_X)$ for rationalization of X}
$$
i.e.
$$
  \mathrm{CE}(\mathfrak{l}(X)) := (A_X,d_X)
  \,.
$$
See Appendix A for more details on rational homotopy theory and Sullivan models.
For example, with this notation then the rationalized spheres $S^n$ are incarnated as
$$
  \mathrm{CE}(\mathfrak{l}S^n)
  =
  \left\{
    \begin{array}{ll}
      ( \mathbb{R}[g_{n}] , d g_n = 0 ) & \mbox{for $n$ odd}
      \\
      ( \mathbb{R}[g_n, g_{2n-1}], d g_n = 0,\; d g_{2n-1} = g_n \wedge g_n ) & \mbox{for $n \gt 0$ even}.
    \end{array}
  \right.
$$
A convenient feature of the dual picture is the following:
if $\mathrm{CE}(\mathfrak{h})\to \mathrm{CE}(\mathfrak{g})$ is a relative Sullivan algebra, that is,
a cofibration in the standard model structure on differential graded commutative algebras (DGCAs),
then the corresponding $L_\infty$-morphism $\mathfrak{g}\to\mathfrak{h}$ is a fibration in the model
 structure whose fibrant objects are $L_\infty$-algebras, due to \cite[prop. 4.36, prop. 4.42]{pridham}. Although relative Sullivan algebras do not exhaust fibrations of $L_\infty$-algebras, they are flexible enough to allow us to realize all the fibrations we will need in the present article as relative Sullivan algebras. See \cite{nLab} for more on the homotopy theory of $L_\infty$-algebras as a category of fibrant objects.

\medskip
The model structure whose fibrant objects are $L_\infty$-algebras in \cite{pridham} is for ordinary $L_\infty$-algebras, not for super $L_\infty$-algebras that we consider here. Nevertheless, the result is readily adapted: A super $L_\infty$-algebras $\mathfrak{g}$
determines a functor $\Lambda \mapsto (\mathfrak{g}\otimes \Lambda)_{\mathrm{even}}$
with values in ordinary $L_\infty$-algebras on the category of finitely generated Grassmann algebras $\Lambda$,
and this construction embeds super $L_\infty$-algebras into this functor category.
(For super Lie algebras this was observed in \cite{Schwarz84}, see \cite{KonechnySchwarz97} and \cite[Cor. 3.3]{Sachse08}).
Now, by  \cite[Theorem 4.35]{pridham}, the opposite model structure for ordinary
$L_\infty$-algebras is cofibrantly generated, and so a standard argument
\cite[section 11.6]{Hirschhorn} gives that this functor category
inherits the corresponding projective model structure.  That is the model structure in which the
computations in this paper take place. However, we need to invoke only a bare minimum of model category theory;
all we use is the computation of homotopy fibers as ordinary fibers of fibration resolutions.
In the following we will find it very useful to succinctly capture results via (commuting) diagrams.
We will use the notation ${\rm hofib}(\phi)$ to indicate the homotopy fiber of a morphism
$\phi$.

\medskip

The spacetimes that we consider now are extended flat superspaces.
(All constructions here globalize from these local models to curved superspacetime by a theory of higher Cartan geometry,
see \cite{StructureTheory} and the references given there.)
Super Minkowski spacetime $\R^{d-1,1\vert N}$ may be identified with its super Lie algebra of (super-)translations. Via the super DG-Lie algebras/super DG-commutative algebras duality, it corresponds to the super DGCA
(differential graded commutative algebra) ${\rm CE}(\R^{d-1,1\vert N})$ which is the
super-DGCA generated by elements $\{e^a\}$ 
of degree $(1, {\rm even})$ and elements
 $\{\psi^a\}$ of degree $(1, {\rm odd})$. The 
 action of the differential is
 given as
 \bea
 d_{{}_{\rm CE}} e^a&=&
 \overline{\psi} \Gamma^a \psi\;,
 \\
 d_{{}_{\rm CE}} \psi &=& 0\;,
 \eea
where $\overline{\psi}$ is the conjugate spinor (whenever defined, depending on dimension).
Geometrically, these generators may be identified with the left invariant 1-forms on super Minkowski spacetime.
We will take appropriate values of $N$ depending on our theories, namely $N = \mathbf{32}$ for M-theory
and $N= \mathbf{16} + \overline{\mathbf{16}}$ for non-chiral type IIA superstring theory. For details and references
we refer to \cite[Section 4]{Bouq}.

\medskip
For every $p\geq 0$ one has a distinguished element $\mu_{p+2}$ in the Chevalley-Eilenberg
algebra ${\rm CE}(\R^{d-1,1\vert N})$, given by
$$
\mu_{p+2}
  :=
 c \overline{\psi} \wedge \Gamma^{a_1 \cdots a_p} \psi \wedge e_{a_1}
\wedge \cdots \wedge e_{a_p}\,,
$$
where $c = 1$ if $(-1)^{p(p-1)/2}$ is even, and $c = i$ otherwise.


\medskip
The organization of the paper is very simple.
In Sec. \ref{Sec M} we discuss the unified supercocycles in M-theory,
then reduce these to type IIA supercocycles in Sec. \ref{Sec red}.
We connect the result to the traditional incarnation of the D-brane cocycles in Sec. \ref{Sec IIA}. In two short Appendices, to make the article more self-contained, we recall a few basic notions from rational homotopy theory and the spinor conventions used in the present article.

\section{The supercocycles in M-theory}
\label{Sec M}

We consider now the cocycles in the brane bouquet \cite{Bouq} that give the WZW term of the Green-Schwarz sigma
model for the M2-brane and the M5-brane, defined on the extended super Minkowski spacetime induced from the cocycle for
the M2-brane. Then we recall \cite{FSS15} how it descends down to 11-dimensional super-Minkowski spacetime
itself, unifying with the M2-cocycle to one single cocycle, but now with coefficient in the rational 4-sphere.

\medskip
The key algebraic fact that governs the M2/M5-brane is the following statement about the
elements $\mu_{p+2}$ from above:
 \begin{prop}[{\cite[(3.26)]{DF}}]
 \label{prop41}
 The elements $\mu_4$ and $\mu_7$ in ${\rm CE}(\R^{10, 1\vert {\bf 32}})$ satisfy
 $$
 d\mu_4=0\;, \qquad d\mu_7=15 \, \mu_4 \wedge \mu_4\;.
 $$
 \end{prop}
 \begin{remark}
The statement of Prop. \ref{prop41} has been rediscovered, in its equivalent incarnation
given below in Corollary \ref{M5cocycleOnM2braneExtension},
in various places, including \cite{BLNPST} and \cite[(8.8)]{CAIB99}, where it was understood as giving the WZW term of the
Green-Schwarz sigma model for the M5-brane on the extended super Minkowski spacetime induced by the WZW term of the M4-brane.
Our Proposition \ref{prop43} below says that this stagewise incarnation of the M5-cocycle on the extension defined by the
M2-cocycle descends to one unified cocycle with coefficients in the rational 4-sphere.

\end{remark}
In terms of $L_\infty$-algebras Proposition \ref{prop41} says the following:
\begin{cor}
\label{mu4-mu7IsS4Cocycle}
The pair $(\mu_4,\mu_7)$ equivalently constitutes the components of an $L_\infty$-morphism
\[
  (\mu_4, \mu_7) : \R^{10, 1\vert {\bf 32}} \to \mathfrak{l}(S^4)\;,
\]
namely, dually, the components of a DG-algebra homomorphism
\begin{align*}
\mathrm{CE}(\mathfrak{l}S^4)&\to {\rm CE}(\R^{10, 1\vert {\bf 32}})\\
g_4&\mapsto \mu_4\\
g_7&\mapsto \tfrac{1}{15}\mu_7\;.
\end{align*}
\end{cor}

  \medskip
Next, we show that the morphism $(\mu_4,\mu_7)\colon \R^{10, 1\vert {\bf 32}}\to \mathfrak{l} S^4$ is actually induced by an equivariant 7-cocycle on the $\frak{m}2\frak{brane}$ extension of the super-Minkovski space $\R^{10, 1\vert {\bf 32}}$. To begin with, the fact that $\mu_4$ is a cocycle, i.e. $d\mu_4=0$, means that $\mu_4$ defines a super-DGCA
morphism
\begin{align*}
\R[g_4]&\to{\rm CE}(\R^{10, 1\vert {\bf 32}})\\
g_4&\mapsto \mu_4\;.
\end{align*}
Consequently, $\mu_4$ defines a morphism of super-$L_\infty$ algebras (which we will simply denote by
the same symbol $\mu_4$)
\[
  \mu_4 : \R^{10, 1\vert {\bf 32}} \to b^3\R\;.
\]
In other words $\mu_4$ is a 4-cocycle on the super-Minkowski space $\R^{10, 1\vert {\bf 32}}$.

\begin{defn}[{\cite[p. 12, p. 16]{Bouq}}]
Write $\frak{m}2\frak{brane}$ for the super $L_\infty$-algebra
which is the homotopy fiber of $\mu_4$, i.e. sitting in a homotopy pullback diagram of the form
$$
\xymatrix{
\frak{m}2\frak{brane} \ar[r]\ar[d] & 0\ar[d] \\
\R^{10,1\vert {\bf 32}} \ar[r]^{\mu_4} & b^3 \R\;.
}
$$
From this description we see that $\frak{m}2\frak{brane}$ is a principal $b^2 \R$-bundle
over $\R^{10,1\vert {\bf 32}}$. In the dual Chevalley-Eilenberg picture, ${\rm CE}(\frak{m}2\frak{brane})$
is obtained from ${\rm CE}(\R^{10,1\vert {\bf 32}})$ by adding a single generator in degree 3 which is
a primitive for $-\mu_4$ (\cite[Prop. 3.5]{Bouq}):
$$
{\rm CE}(\frak{m}2\frak{brane})=\big({\rm CE}(\R^{10,1\vert {\bf 32}}) \otimes \R[h_3]\;;
~ dh_3=-\mu_4\big)\;,
$$
with $\deg h_3=3$. The dual morphism is simply the obvious inclusion
${\rm CE}(\R^{10,1\vert {\bf 32}})\hookrightarrow {\rm CE}(\frak{m}2\frak{brane})$.
\label{def65}
\end{defn}
This definition allows expanding the relation $d\mu_7=15 \mu_4 \wedge \mu_4$
from Prop. \ref{prop41} as
follows.

\begin{cor}
\label{M5cocycleOnM2braneExtension}
There is a super $L_\infty$-cocycle of the form
$$
\xymatrix{
\frak{m}2\frak{brane} \ar[rrr]^-{h_3 \wedge \mu_4 + \tfrac{1}{15}\mu_7} &&& \mathfrak{l} (S^{7}) = b^6\R
}\;.
$$
\label{cor8}
\end{cor}
Explicitly, the above corollary precisely says that the element $h_3 \wedge \mu_4 + \tfrac{1}{15}\mu_7$ in
${\rm CE}(\frak{m}2\frak{brane})$ is closed, which is immediate:
\bea
d(h_3 \wedge \mu_4 + \tfrac{1}{15}\mu_7)&=&dh_3\wedge \mu_4-h_3\wedge d\mu_4+\tfrac{1}{15}d\mu_7
\nonumber\\
&=&-\mu_4\wedge\mu_4+\tfrac{1}{15}d\mu_7
\nonumber\\
&=&0\;.
\eea
By the defining characterization both $\frak{m}2\frak{brane}$ and $\mathfrak{l} S^{7}$ are naturally $b^2 \R$-principal bundles
according to \cite{NSS12}, so it is natural to
ask whether $h_3 \wedge \mu_4+ \tfrac{1}{15}\mu_7$ is $b^2 \R$-equivariant. If so
then by the general theory of higher bundles \cite{NSS12} it will descend to a morphism
\[
\R^{10,1\vert {\bf 32}}\to \mathfrak{l} (S^4)
\]
in the homotopy category of super $L_\infty$-algebras, i.e., to a span of $L_\infty$-morphisms of the form
\[
\R^{10,1\vert {\bf 32}}\xleftarrow{\sim}\R_{\rm res}^{10,1\vert {\bf 32}}\to \mathfrak{l} (S^4)\;,
\]
where $\R_{\rm res}^{10,1\vert {\bf 32}}\xrightarrow{\sim} \R^{10,1\vert {\bf 32}}$ is a quasi-isomorphism (we say that $\R_{\rm res}^{10,1\vert {\bf 32}}$ is a resolution of super Minkowski spacetime).
To exhibit in  components the equivariance of the 7-cocycle in Corollary \ref{cor8}
with respect to this action we need an explicit  resolution of super Minkowski spacetime:

\begin{defn}[{\cite[Section III]{FSS15}}]
\label{Def-CEgg}
Let $\R_{\rm res}^{10,1\vert {\bf 32}}$ be the super $L_\infty$-algebra whose
 Chevalley-Eilenberg algebra is obtained from $\mathrm{CE}(\R^{10,1\vert {\bf 32}})$ by adding two generators $h_3$ and $g_4$, of degree 3 and 4, respectively, with $dh_3=g_4 - \mu_4$ and $dg_4=0$:
$$
{\rm CE}(\R_{\rm res}^{10,1\vert {\bf 32}}):=
\Big({\rm CE}(\R^{10,1\vert {\bf 32}})\otimes\R[h_3, g_4]\;;
~
dh_3=g_4 - \mu_4\;, ~dg_4=0
  \Big)\;.
$$
\end{defn}

\begin{prop}
\label{resolution-one}
The canonical morphism $\R_{\rm res}^{10,1\vert {\bf 32}} \overset{\simeq}{\longrightarrow}
\R^{10,1\vert {\bf 32}}$, dual to the obvious inclusion  ${\rm CE}(\R^{10,1\vert {\bf 32}})\hookrightarrow {\rm CE}(\R_{\rm res}^{10,1\vert {\bf 32}})$, is
an equivalence of $L_\infty$-algebras. It factors the natural projection
$\frak{m}2\frak{brane} \to
\R^{10,1\vert {\bf 32}}$ through the morphism
$\frak{m}2\frak{brane} \to \R_{\rm res}^{10,1\vert {\bf 32}}$ whose dual map is
\begin{align*}
{\rm CE}(\R_{\rm res}^{10,1\vert {\bf 32}})&\longrightarrow {\rm CE}(\frak{m}2\frak{brane})\\
h_3 &\longmapsto h_3\\
g_4 &\longmapsto 0
\end{align*}
and the identity on all other generators.
\end{prop}
\begin{proof}
The only nontrivial part in the statement is the homotopy equivalence
$\R_{\rm res}^{10,1\vert {\bf 32}} \overset{\simeq}{\longrightarrow}
\R^{10,1\vert {\bf 32}}$. In terms of Chevalley-Eilenberg algebras, this amounts to saying that the quotient algebra
\[
{\rm CE}(\R_{\rm res}^{10,1\vert {\bf 32}})/{\rm CE}(\R^{10,1\vert {\bf 32}})
\]
is homotopy equivalent to $\R$ as a DGCA. One has
\bea
{\rm CE}(\R_{\rm res}^{10,1\vert {\bf 32}})/{\rm CE}(\R^{10,1\vert {\bf 32}})
&\cong &
\left(\R[h_3, g_4]\;; ~ dh_3=g_4\;, ~ dg_4=0\right)
\nonumber\\
&=&\R[h_3, dh_3]
\nonumber\\
&\cong & \R\;,
\eea
where the last quasi-isomorphism is the evident one.
\end{proof}
\begin{prop}\label{equivalent-cocycle}
The 4-cocycle $\tilde{\mu}_4\colon \R_{\rm res}^{10,1\vert {\bf 32}}\to b^3 \R$, dual to the obvious inclusion $\R[g_4]\to {\rm CE}(\R_{\rm res}^{10,1\vert {\bf 32}})$, fits into a homotopy commutative diagram of super $L_\infty$-algebras
\[
\xymatrix{ \R_{\rm res}^{10,1\vert {\bf 32}}\ar[r]^{\simeq}\ar[dr]_{\tilde{\mu}_4}&  \R^{10,1\vert {\bf 32}}\ar[d]^{\mu_4}\\
&b^3 \R \;.
}
\]
\end{prop}
\begin{proof}
In the dual Chevalley-Eilenberg picture we have to show that the diagram
\[
\xymatrix{
\Big({{\rm CE}(\R^{10,1\vert {\bf 32}})\otimes\R[h_3, g_4]\;;\
\atop{
dh_3=g_4 - \mu_4\;, ~dg_4=0}}
  \Big)&&  \mathrm{CE}(\R^{10,1\vert {\bf 32}})\ar[ll]_-{\simeq}\\
&&(\R[g_4];\, dg_4=0)\ar[ull]^-{g_4\mapsto g_4}\ar[u]_{g_4\mapsto \mu_4}
}
\]
is homotopy commutative. For this is it sufficient that there exists a morphism
\[
\xi\colon \big(\R[g_4];\, dg_4=0\big)\to
\Big({{\rm CE}(\R^{10,1\vert {\bf 32}})\otimes\R[h_3, g_4]\;;\
\atop{
dh_3=g_4 - \mu_4, \;\; dg_4=0}}
  \Big)\otimes \Omega^\bullet(\Delta^1)\;,
\]
where $\Omega^\bullet(\Delta^1)=\R[t,dt]$ is the DGCA of polynomial differential forms on the 1-simplex
$\Delta^1$ (this does give a path space object for right homotopies according to \cite[Lemma 4.32]{pridham}), such that $\xi$ evaluated at $0$ maps $g_4$ to $g_4$ while  $\xi$ evaluated at $1$ maps $g_4$ to $\mu_4$.  In other words, we are looking for a closed degree 4 element $\xi_4(t)+\xi_3(t)dt$ in ${\rm CE}(\R^{10,1\vert {\bf 32}}_{\mathrm{res}})\otimes \Omega^\bullet(\Delta^1)$ with $\xi_4(0)=g_4$ and $\xi_4(1)=\mu_4$. An obvious choice is
\[
\left(g_4+t(\mu_4-g_4)\right)+h_3\,dt\;.
\]

\vspace{-4mm}
\end{proof}

Now we have all the ingredients to complete our diagram:
\begin{prop}{\cite[Section III]{FSS15}}
\label{prop43}
Starting with the cocycle $\mu_4$, there is commutative diagram of $L_\infty$-algebras of the form
$$
\xymatrix{
& \frak{m}2\frak{brane} \ar[dd]\ar[rrd] \ar[rrrr]^-{h_3 \wedge \mu_4+ \tfrac{1}{15}\mu_7}
&&&& \mathfrak{l} (S^{7}) \ar[dd]\ar[lld]
\\
& & &0\ar[dd]|(.38)\hole|(.41)\hole|(.43)\hole
 \\
\R^{10,1\vert {\bf 32}} & \R_{\rm res}^{10,1\vert {\bf 32}} \ar[l]_{\simeq} \ar[drr]_{\tilde{\mu}_4}
\ar[rrrr]|(.51)\hole|(.53)\hole|(.56)\hole^-{h_3 \wedge (g_4+\mu_4)+ \tfrac{1}{15}\mu_7}
&&&& \mathfrak{l} (S^{4}) \ar[dll] \\
&&& b^3 \R &&
}
$$
where the two front faces of the prism are homotopy pullbacks.
\end{prop}

\begin{proof}
We have to check that the diagram exists and commutes at the level of the dual CE-algebras. Forgetting the differentials, this is immediate in terms of the defining generators:
each generator is mapped to the generator of the same name, if present, in the codomain, or to zero otherwise,
except for $g_7\in {\rm CE}(\mathfrak{l}S^{7})$ which is sent to $h_3 \wedge \mu_4 +\tfrac{1}{15}\mu_7$,
and $g_7 \in {\rm CE}(\mathfrak{l}(S^{4}))$, which is sent to $h_3 \wedge (g_4+\mu_4)+ \tfrac{1}{15}\mu_7$,
as indicated. When the differentials are taken into account, the only thing to be checked is that the middle
horizontal map respects the CE-differentials. Indeed, by Proposition \ref{prop41}, we have
\bea
d\big(h_3 \wedge (g_4+\mu_4)+ \tfrac{1}{15}\mu_7\big) &=& (g_4 -\mu_4) \wedge
(g_4 + \mu_4) + \mu_4 \wedge \mu_4
\nonumber\\
&=& g_4 \wedge g_4\;.
\eea
By Definition \ref{Def-CEgg}, this says that indeed
$g_7 \mapsto h_3 \wedge (g_4+\mu_4)+ \tfrac{1}{15}\mu_7$ respects the CE-differential.
\end{proof}

We conclude this section by showing that, as anticipated, the morphism of $L_\infty$-algebras exhibited in Proposition \ref{prop43} is equivalent to that in Corollary \ref{cor8}. That is, we have the following.
\begin{prop}
\label{EquivalenceOfTwoIncarnationsOfMBraneCocycles}
There is a homotopy commutative diagram of $L_\infty$-algebra morphisms
\[
\xymatrix{ \R_{\rm res}^{10,1\vert {\bf 32}}\ar[rrr]^{\simeq}
\ar[drrr]_{\hspace{-7mm}\small{h_3 \wedge (g_4+\mu_4)+ \tfrac{1}{15}\mu_7}}
&&&
 \R^{10,1\vert {\bf 32}}\ar[d]^{(\mu_4,\tfrac{1}{15}\mu_7)}\\
&&&\mathfrak{l} (S^{4})\;.
}
\]
\end{prop}
\begin{proof}
We have to show that the dual diagram of DGCAs
\[
\xymatrix{
\Big({{\rm CE}(\R^{10,1\vert {\bf 32}})\otimes\R[h_3, g_4]\;;\
\atop{
dh_3=g_4 - \mu_4, ~dg_4=0}}
  \Big)
  &&
  \mathrm{CE}(\R^{10,1\vert {\bf 32}})\ar[ll]_-{\simeq}\\
  \\
&&(\R[g_4,g_7];\, dg_4=0,~ dg_7=g_4\wedge g_4)
\ar[uull]^{\hspace{-20mm} g_4\mapsto g_4, \atop{ g_7\mapsto h_3 \wedge (g_4+\mu_4)+
\tfrac{1}{15}\mu_7}}\ar[uu]_{g_4\mapsto \mu_4, \atop{ g_7\mapsto \tfrac{1}{15}\mu_7}}
}
\]
is homotopy commutative. Reasoning as in the proof of Proposition \ref{equivalent-cocycle}, we have to exhibit a closed degree 4 element $\xi_4(t)+\xi_3(t)dt$ and a degree 7 element $\xi_7(t)+\xi_6(t)dt$ in ${\rm CE}(\R^{10,1\vert {\bf 32}}_{\mathrm{res}})\otimes \Omega^\bullet(\Delta^1)$ such that $d(\xi_7(t)+\xi_6(t)dt)=(\xi_4(t)+\xi_3(t)dt)\wedge(\xi_4(t)+\xi_3(t)dt)$,
with $\xi_4(0)=g_4$, $\xi_4(1)=\mu_4$, $\xi_7(0)=h_3 \wedge (g_4+\mu_4)+ \tfrac{1}{15}\mu_7$ and $\xi_7(1)=\frac{1}{15}\mu_7$. An immediate  choice is
\begin{align*}
\xi_4(t)+\xi_3(t)dt&=\left(g_4+t(\mu_4-g_4)\right)+h_3\,dt\;,
\\
\xi_7(t)+\xi_6(t)dt&=(1-t)h_3 \wedge ((1+t)\mu_4+(1-t)g_4)+ \tfrac{1}{15}\mu_7\;.
\end{align*}
Indeed, with this choice the boundary conditions on the $\xi_i$'s are trivially satisfied and, moreover, we have
\begin{align*}
d(\xi_7(t)+\xi_6(t)dt)&=d\left((1-t)h_3 \wedge ((1+t)\mu_4+(1-t)g_4)+ \tfrac{1}{15}\mu_7\right)\\
&=(-dt\, h_3+(1-t)dh_3)\wedge  ((1+t)\mu_4+(1-t)g_4)\\
&\qquad -(1-t)h_3 \wedge (dt \, \mu_4 +(1+t)d\mu_4 -dt\, g_4 +(1-t)dg_4) +\tfrac{1}{15}d\mu_7\\
&=(h_3\,dt+(1-t)g_4 - (1-t)\mu_4))\wedge  ((1+t)\mu_4+(1-t)g_4)\\
&\qquad -(1-t)h_3 \wedge (dt \, \mu_4  -dt\, g_4) +\mu_4\wedge\mu_4\\
&=(1+t)\mu_4\wedge h_3\,dt+(1-t)g_4\wedge h_3\,dt+(1-t^2)g_4\wedge \mu_4+(1-t)^2g_4\wedge g_4\\
&\qquad -(1-t^2)\mu_4\wedge\mu_4-(1-t)^2\mu_4\wedge g_4-(1-t)\mu_4\wedge h_3\, dt+(1-t)g_4\wedge h_3\, dt\\
&\qquad +\mu_4\wedge\mu_4\\
&=2t\mu_4\wedge h_3\,dt+2(1-t)g_4\wedge h_3\,dt+2t(1-t)g_4\wedge \mu_4+(1-t)^2g_4\wedge g_4\\
&\qquad +t^2\mu_4\wedge\mu_4
\end{align*}
and
\begin{align*}
(\xi_4(t)+\xi_3(t)dt)\wedge(\xi_4(t)+\xi_3(t)dt)&=\left(\left(g_4+t(\mu_4-g_4)\right)+h_3\,dt\right)\wedge \left(\left(g_4+t(\mu_4-g_4)\right)+h_3\,dt\right)\\
&=g_4\wedge g_4+2t g_4\wedge (\mu_4-g_4)+ t^2 (\mu_4-g_4)\wedge (\mu_4-g_4)\\
&\qquad +2 g_4\wedge h_3\,dt+2t(\mu_4-g_4)\wedge h_3\, dt\\
&=(1-t)^2g_4\wedge g_4+2t (1-t) g_4\wedge \mu_4 +t^2 \mu_4\wedge \mu_4\\
&\qquad +2(1-t) g_4\wedge h_3\,dt+2t\mu_4\wedge h_3\, dt\;.
\end{align*}
\end{proof}

\begin{remark}
 {\bf (i)} The form of the equivariant cocycle in Proposition \ref{prop43} is that of the curvature of the WZW term of the sigma-model describing the M5-brane as considered in \cite{BLNPST}.

\item {\bf (ii)} As the notation suggests, in terms of rational homotopy theory, Proposition \ref{prop43} says that the CE-elements
$\mu_4$ and $\mu_7$ of Proposition \ref{prop41} define a cocycle with values in the rational 4-sphere. In the discussions
in \cite{SS} we see that under Lie integration and globalization in higher Cartan geometry, these elements encode the supergravity C-field and its magnetic dual.
\end{remark}

\section{The dimensional reduction to IIA superstring theory}
\label{Sec red}

We derive now the dimensional reduction of the rational charge structure of M-theory
to that of type IIA string theory, realized as cyclic cohomology in $L_\infty$-homotopy theory. It takes the
(rational) M2/M5-brane charge structure from the previous section to a (rational)
twisted cohomology theory for the F1/D$p$/NS5-branes in type IIA. In fact we find that
super $L_\infty$-theoretically there is an \emph{equivalence} between these two rational
cohomology theories, in 11 and in 10 dimensions.

The (M-theory) super-Minkowski spacetime $\R^{10,1\vert {\bf 32}}$ is rationally a $S^1_{\R}$-principal bundle
over the (type IIA) super-Minkowski spacetime $\mathbb{R}^{9,1\vert \mathbf{16} + \overline{\mathbf{16}}}$.
\begin{prop}[{\cite[prop. 4.5]{Bouq}}]
\label{R11AsExtensionOfR10}
There is a homotopy fiber sequence of $L_\infty$-algebras
\[
\xymatrix{
 \R^{10,1\vert {\bf 32}}
  \ar[rr]\ar[d]
 &&
 0
   \ar[d]
 \\
 \mathbb{R}^{9,1\vert \mathbf{16} + \overline{\mathbf{16}}}
  \ar[rr]_-{\overline{\psi} \Gamma_{10} \psi}
  && b\R\;,
}
\]
which exhibits $\mathbb{R}^{10,1\vert \mathbf{32}}$ as the central extension of super Lie algebras
classified by the 2-cocycle $\overline{\psi}\Gamma_{10}\psi$
(which is the D0-brane cocycle, see def. \ref{IIADpbranecocycles} below).
\end{prop}

 Below in Prop. \ref{Hom-Prop} we show that the double dimensional reduction of the M-brane
 cocycles of Prop. \ref{EquivalenceOfTwoIncarnationsOfMBraneCocycles}
 along the fibration of super Minkowski spacetimes of Prop. \ref{R11AsExtensionOfR10}
 is neatly captured by the following Sullivan model for rational cyclic cohomology:
\begin{prop}
  \label{FreeLoopSpaceSullivanModel}
  Let $X$ be a simply connected topological space whose rationalization admits a minimal Sullivan model $(\wedge^\bullet V, d_X)$.
  Then
  \begin{enumerate}
    \item
      A Sullivan model for the rationalization of the free loop space $\mathcal{L}X$ of $X$ is given by
      $$
        \mathrm{CE}(\mathfrak{l}(\mathcal{L}X)) =  (\wedge^\bullet ( V \oplus s V ),
        d_{{}_{\mathcal{L}X}})
        \,,
      $$
      where $s V$ is $V$ with degrees shifted down by one, and with $d_{{}_{\mathcal{L}}X}$
      acting for $v \in V$ as
      $$
        d_{{}_{\mathcal{L}X}} \, v = d_{{}_X} \, v
        \,,
        \;\;\;\;\;\;
        d_{{}_{\mathcal{L}X}} \, s v = - s d_{{}_X} \, v
        \,,
      $$
      where on the right $s \colon V \to s V$ is extended uniquely as a graded derivation.
   \item
     A Sullivan model for the rationalization of the homotopy quotient $\mathcal{L}X /\!/ S^1$
     (presented by the Borel construction $\mathcal{L}X \times_{S^1} E S^1$) for the canonical
     circle group action on the free loop space (by rotation of loops) is given by
     $$
       \mathrm{CE}(\mathfrak{l}( \mathcal{L}X/\!/S^1 ))
       =
       \big( \wedge^\bullet( V \oplus s V \oplus \langle \omega_2\rangle), d_{{}_{\mathcal{L}X/\!/S^1}} \big)
     $$
     with
     $$
       d_{{}_{\mathcal{L}X/\!/S^1}} \, \omega_2 = 0
     $$
     and with $d_{{}_{\mathcal{L}X/\!/S^1}}$ acting on $w \in \wedge^\bullet V \oplus sV $ as
     $$
       d_{{}_{\mathcal{L}X/\!/S^1}}\, w = d_{{}_{\mathcal{L}X}}\, w + \omega_2 \wedge s w
       \,.
     $$
     Moreover, the canonical sequence of $L_\infty$-homomorphisms
     $$
       \xymatrix{
         \mathfrak{l}(\mathcal{L}X)
           \ar[r]
         & \mathfrak{l}(\mathcal{L}X/\!/S^1)
           \ar[r]
         & b \mathbb{R}
       }
     $$
     is a rational model for the homotopy fiber sequence
     $$
       \xymatrix{
         \mathcal{L}X
         \ar[r]
         &
         \mathcal{L}X/\!/S^1
         \ar[r]
         &
         B S^1
       }
     $$
     that exhibits the homotopy quotient.
  \end{enumerate}
\end{prop}
The first statement is due to \cite{VigueSullivan}, the second due to \cite{VigueBurghelea}.

Here we are concerned with the following special case of this fact:
\begin{example}
Let $X = S^4$ be the 4-sphere, with
$$
  \mathrm{CE}(\mathfrak{l}(S^4)) = \big(\wedge^\bullet \langle g_4, g_7\rangle,
 ~  d_{{}_{S^4}}\, g_4 = 0\,, ~d_{{}_{S^4}}\, g_7 = - \tfrac{1}{2}g_4 \wedge g_4\big)
  \,.
$$
(We have rescaled the generator $g_7$ by a factor of $-\tfrac{1}{2}$ with respect to the conventions in the previous section;
this yields an isomorphic model but serves to reduce prefactors in the following formulas.)
By Prop. \ref{FreeLoopSpaceSullivanModel} the free loop space of $S^4$ is modeled by
 \[
\mathrm{CE}(\mathfrak{l}({\mathcal{L}S^4}))=\Bigl( \mathbb{R}[\omega_4,\omega_6,h_3, h_7]\,;
\begin{array}{l}
d_{{}_{\mathcal{L}S^4}}\; \omega_4= 0, ~~d_{{}_{\mathcal{L}S^4}}\; \omega_6 = h_3\wedge \omega_4,
\\
d_{{}_{\mathcal{L}S^4}}\; h_3=0, ~~~d_{{}_{\mathcal{L}S^4}}\; h_7 = -\tfrac{1}{2}\omega_4\wedge \omega_4
\end{array}
\Bigr)\;.
\]
and the homotopy quotient by $S^1$ of the free loop space of $S^4$ is modeled as
\[
\mathrm{CE}(\mathfrak{l}({\mathcal{L}S^4}/\!/S^1))=\Bigl( \mathbb{R}[\omega_2,\omega_4,\omega_6,h_3, h_7]\,;
\begin{array}{l}
d_{{}_{\mathcal{L}S^4/S^1}}\; \omega_2=0, ~d_{{}_{\mathcal{L}S^4/S^1}}\; \omega_4=h_3\wedge \omega_2, ~d_{{}_{\mathcal{L}S^4/S^1}}\; \omega_6=h_3\wedge \omega_4,
\\
d_{{}_{\mathcal{L}S^4/S^1}}\; h_3=0, ~~~d_{{}_{\mathcal{L}S^4/S^1}}\; h_7= -\tfrac{1}{2}\omega_4\wedge \omega_4 + \omega_6\wedge \omega_2
\end{array}
\Bigr)\;.
\]
\end{example}
The following lemma is then immediate.
\begin{lemma}\label{s1-quotient}
We have a homotopy pushout diagram of DGCAs
\[
\xymatrix{
\bigl(\mathbb{R}[\omega_2],d\omega_2=0\bigr)\ar[r]\ar[d]&
\left({
    \R[ \omega_2,\omega_4,\omega_6,h_3,h_7] \;;\atop{\,dh_3=0,  ~d\omega_2=0, ~d\omega_4=h_3\wedge \omega_2, ~d\omega_6=h_3\wedge \omega_4, ~dh_7= -\tfrac{1}{2}\omega_4 \wedge \omega_4+\omega_6\wedge \omega_2}}
  \right)\ar[d]\\
\mathbb{R}\ar[r]& \left({
    \R[\omega_4,\omega_6,h_3, h_7] \;;\atop{\,dh_3=0, ~ d\omega_4=0, ~d\omega_6=h_3\wedge \omega_4,
    ~ dh_7= -\tfrac{1}{2}\omega_4\wedge \omega_4}}
  \right)
}
\]
inducing the homotopy fiber sequence
\[
\xymatrix{
  \mathfrak{l}(\mathcal{L}S^4)
    \ar[r] &
  \mathfrak{l}({\mathcal{L}S^4}/\!/S^1) \ar[r]^-{\omega_2} &
 b \R
}
\]
of $L_\infty$-algebras.
\end{lemma}

\begin{remark}
The idea that double dimensional reduction in string theory is mathematically formalized by looping of cocycles
has been considered in \cite{AE, BV, MathaiSati, S, S0}. Here we will only be concerned with the dimensional reduction of a cocycle
 of the form $\mathbb{R}^{10,1\vert \mathbf{32}} \to \mathfrak{l}(S^4)$ along the projection  $\mathbb{R}^{10,1\vert \mathbf{32}} \to \mathbb{R}^{9,1\vert \mathbf{16} + \overline{\mathbf{16}}}$, but this is just a particular instance of a general procedure \cite{TFolds},
 as illustrated in \cite{S1}.
\end{remark}

To dimensionally reduce our morphism  $\mathbb{R}^{10,1\vert \mathbf{32}} \to \mathfrak{l} S^4$ to a morphism
$\mathbb{R}^{9,1\vert \mathbf{16} + \overline{\mathbf{16}}}\to
\mathfrak{l}({\mathcal{L}S^4}/\!/S^1)$ we systematically ``isolate'' the vertical coordinate $e^{11}$ in $\mathbb{R}^{10,1\vert \mathbf{32}}$. Before doing so, it will be useful to recall the definition of the $F1$-brane cocycle on  $\mathbb{R}^{9,1\vert \mathbf{16} + \overline{\mathbf{16}}}$ as it will show up in the dimensional reduction of the cocycle $(\mu_4,\mu_7)\colon \mathbb{R}^{10,1\vert \mathbf{32}} \to \mathfrak{l}(S^4)$.

\begin{defn}[{\cite{CAIB99}, \cite[Def. 4.2]{Bouq}}  ]
\label{def.mu-f1}
The type IIA superstring (or \emph{F1}-brane) cocycle is the super Lie algebra 3-cocycle
\[
\mu_{{}_{F1}}:= i \overline{\psi} \Gamma^a \Gamma_{10}  \psi \wedge e_{a}~\colon
~~\mathbb{R}^{9,1\vert \mathbf{16} + \overline{\mathbf{16}}}\longrightarrow b^2\mathbb{R}
 \,.
\]
\end{defn}

\begin{remark}As the name indicates, the 3-cocycle $\mu_{{}_{F1}}$ plays a relevant role in type IIA superstring theory. We are going to discuss this in the next section. In particular, the cocycle $\mu_{{}_{F1}}$ is the
one which gives rise to the WZW term of the Green-Schwarz sigma model for
the type IIA string (see \cite{Bouq} and references therein).
\end{remark}

The main statement of this section is now the following:
\begin{prop}
\label{Hom-Prop}
There is a canonical \emph{dimensional reduction} isomorphism of hom-sets
\[
  \mathrm{Hom}_{L_\infty}(\mathbb{R}^{10,1\vert \mathbf{32}},\mathfrak{l} S^4)
    \xrightarrow{\cong}
  \mathrm{Hom}_{L_\infty/b \R}
  \big(\mathbb{R}^{9,1\vert \mathbf{16} + \overline{\mathbf{16}}},\mathfrak{l}({\mathcal{L}S^4}/\!/S^1)\big)
  \,,
\]
where on the right we have $L_\infty$-morphisms over $b \R$, via Proposition \ref{R11AsExtensionOfR10}
and Lemma \ref{s1-quotient}.
Moreover, the image under this isomorphism of any $L_\infty$-morphism
\[
  \mathbb{R}^{10,1\vert \mathbf{32}} \xrightarrow{(g_4,~g_7)} \mathfrak{l} S^4
\]
is the $L_\infty$-morphism over $b\R$
\[
  \xymatrix{
    \mathbb{R}^{9,1\vert \mathbf{16} + \overline{\mathbf{16}}}
     \ar[drr]_-{\overline{\psi}\Gamma_{10}\psi}
     \ar[rrrr]^{(\omega_2,~ \omega_4, ~\omega_6, ~h_3, ~h_7)}
     &&&&
    \mathfrak{l}({\mathcal{L}S^4}/\!/S^1)\;,
    \ar[dll]^-{\omega_2}
    \\
    && b \R
  }
\]
where $\omega_2$ in the 5-tuple $(\omega_2, \omega_4, \omega_6, h_3, h_7)$ is the the D0-brane cocycle  $\overline{\psi}\Gamma_{10}\psi\in \mathrm{CE}(\mathbb{R}^{9,1\vert \mathbf{16} + \overline{\mathbf{16}}})$
from Proposition \ref{R11AsExtensionOfR10} and where $h_3,h_7, \omega_4, \omega_6 \in \mathrm{CE}(\mathbb{R}^{9,1\vert \mathbf{16} + \overline{\mathbf{16}}})$ are uniquely defined by the relations
$$
  \omega_4 = g_4 - h_3 \wedge e^{11}
  \;\;\;
  \mbox{and}
  \;\;\;
  h_7 = g_7 + \omega_6 \wedge e^{11}
  \,.
$$
In the case that $g_4 = \mu_4$ and $g_7 = \mu_7$ as we had  in Corollary \ref{mu4-mu7IsS4Cocycle}, the element
 $ h_3$
is the F1-brane cocycle $\mu_{{}_{F1}} := \overline{\psi} \Gamma^a \Gamma_{10}\psi \wedge e_a$ from Definition \ref{def.mu-f1}.
\end{prop}
\begin{proof}
By the fact that the underlying graded algebras are free, and since $e^{11}$ is a generator of odd degree,
the given decomposition for $\omega_4$ and $h_7$ is unique.
Hence it is sufficient to observe that under this decomposition the defining equations
$$
  d g_4 = 0
  \,,\;\;\;
  d g_{7} = -\tfrac{1}{2} g_4 \wedge g_4
$$
for the $\mathfrak{l}S^4$-valued cocycle on $\R^{10,1\vert \mathbf{32}}$
turn into the equations for an $\mathfrak{l} ( \mathcal{L}S^4 /\!/S^1)$-valued cocycle on
$\R^{9,1\vert \mathbf{16} + \overline{\mathbf{16}}}$. This is straightforward (using
$d_{{}_{\R^{10,1\vert \mathbf{32}}}} \, e^{11} = \overline{\psi}\Gamma_{10}\psi = \omega_2$):
$$
  \begin{aligned}
    & d_{{}_{\R^{10,1\vert \mathbf{32}}}}\, ( \omega_4 + h_3 \wedge e^{11} ) = 0
    \\
    \Leftrightarrow \;\;\;\;
    &
    d_{{}_{\R^{9,1\vert \mathbf{16}+ \overline{\mathbf{16}}}}}\, (\omega_4)  - h_3 \wedge
    d_{{}_{\R^{9,1\vert \mathbf{16}+ \overline{\mathbf{16}}}}}\,e^{11} = 0
    \;\;\;\; \mbox{and} \;\;\;\;
    d_{{}_{\R^{9,1\vert \mathbf{16}+ \overline{\mathbf{16}}}}}\, h_3 = 0
    \\
    \Leftrightarrow \;\;\;\;
    &
    d_{{}_{\R^{9,1\vert \mathbf{16}+ \overline{\mathbf{16}}}}}\, \omega_4 = h_3 \wedge \omega_2
    \qquad \qquad \qquad \qquad \; \mbox{and} \;\;\;\;
    d_{{}_{\R^{9,1\vert \mathbf{16}+ \overline{\mathbf{16}}}}}\, h_3 = 0\;,
  \end{aligned}
$$

as well as

$$
  \begin{aligned}
    & d_{{}_{\R^{10,1\vert \mathbf{32}}}}\, ( h_7 - \omega_6 \wedge e^{11} )
      = -\tfrac{1}{2}( \omega_4 + h_3 \wedge e^{11} ) \wedge (\omega_4 + h_3\wedge e^{11})
   \\
    \Leftrightarrow \;\;\;\;
    &
    d_{{}_{\R^{9,1\vert \mathbf{16}+ \overline{\mathbf{16}}}}}\, h_7 - \omega_6 \wedge \omega_2
    =
    -\tfrac{1}{2}\omega_4 \wedge \omega_4
    \;\;\;\; \mbox{and} \;\;\;\;
    -  d_{{}_{\R^{9,1\vert \mathbf{16}+ \overline{\mathbf{16}}}}}\, \omega_6 = - h_3 \wedge \omega_4
    \\
    \Leftrightarrow \;\;\;\;
    &
    d_{{}_{\R^{9,1\vert \mathbf{16}+ \overline{\mathbf{16}}}}}\, h_7 = -\tfrac{1}{2}\omega_4 \wedge \omega_4  + \omega_6 \wedge \omega_2
    \;\;\;\; \mbox{and} \qquad \;
    d_{{}_{\R^{9,1\vert \mathbf{16}+ \overline{\mathbf{16}}}}}\, \omega_6 = h_3 \wedge \omega_4\;.
  \end{aligned}
$$
\end{proof}

\begin{remark}
[Self-duality]\label{SelfDuality}

The fields of (super)gravity and the 11d supergravity equations of motion are implied by
a torsion-free globalization over an 11-dimensional super-spacetime $X$ of the supercocycles for the M2-brane, see \cite{StructureTheory}.
Moreover, globalizing the combined M2/M5-brane cocycles with coefficients in $\mathfrak{l}S^4$ yields the
BPS-charge extensions of the superisometry superalgebras of supergravity solutions, see \cite{SS}.
If we write $G_4$ and $G_7$ for the corresponding super-differential forms on spacetime, i.e. for the components
$$
  (G_4, G_7)
    \;:\;
  \mathrm{CE}(\mathfrak{l}S^4)
    \longrightarrow
  \Omega^\bullet(X)
$$
then these supergravity equation of motion imply in particular that $G_7$ is the Hodge dual of $G_4$
$$
  G_7 = \star_{{}_{11}} G_4
  \,.
$$
In terms of these global forms the transmutation of Proposition \ref{Hom-Prop}
corresponds to the Gysin-sequence decomposition of \cite[section 4.2]{MathaiSati}
 $$
  G_4 =  R_4 + H_3 \wedge e^{11}\;, \qquad
  G_7 = H_7 - R_6 \wedge e^{11}
$$
with
$$
 d R_4 + H_3 \wedge R_2 = 0\;, \qquad
  d R_6 = 2 H_3 \wedge R_4\;, \qquad
  d H_7 = -\tfrac{1}{2}R_4 \wedge R_4 + R_6 \wedge R_2 \;.
$$
The above are local versions of the derivations of the field strengths and their Bianchi identities
for the D$p$-branes (for $p \in \{0,2,4\}$) of type IIA string theory;
that last equation is the Chern-Simons term of the NS5-brane
(see \cite{FS} \cite{NS5} and references therein).
The 11d sugra equation of motion $G_7 = \star_{{}_{11}} G_4$ then implies
$$
 R_4 = \star_{{}_{10}} R_6\;, \qquad
 H_3 = \star_{{}_{10}} H_7\;,
$$
as it should be.
\end{remark}

\section{The supercocycles in type IIA superstring theory}
\label{Sec IIA}

In the previous section we have obtained the unified charge structure of F1/D$p$/NS5-branes
by dimensional reduction of the unified charge structure of M2/M5-branes in degree 4-cohomotopy,
as a statement in super $L_\infty$-cohomology theory. But traditionally the super-cocycles
for the WZW terms of the D$p$-branes are considered separately (for each $p$) as cocycles
on extended super Minkowski spacetime (see \cite{Bouq} for literature and $L_\infty$-theoretic formulation).
Here we show that the same $L_\infty$-descent mechanism which unifies the M2-brane charges with the
M5-brane charges (in Sec. \ref{Sec M}) also unifies these separate D$p$-brane charges to a single
charge in (the rational image of) twisted K-theory.

\medskip
Recall from Def. \ref{def.mu-f1} that the superstring (or \emph{F1}-brane) cocycle in type IIA
is the super Lie algebra 3-cocycle
$$
\mu_{{}_{F1}}
  :=
   i \overline{\psi} \Gamma^a \Gamma_{10} \psi \wedge e_{a}~\colon~
\mathbb{R}^{9,1\vert \mathbf{16} + \overline{\mathbf{16}}}\longrightarrow b^2\R\;.
$$
We will now put this cocycle to work.

\begin{defn}[{\cite[Def. 4.2]{Bouq}}]
The super Lie 2-algebra $\mathfrak{string}_{\mathrm{IIA}}$ is the the super Lie 2-algebra extension of $ \mathbb{R}^{9,1\vert \mathbf{16}+\overline{\mathbf{16}}}$ classified by the 3-cocycle $\mu_{{}_{F1}}$. Equivalently, it is the homotopy fiber
(in super $L_\infty$-algebras) of the 3-cocycle $\mu_{{}_{F1}}$:
\[
\xymatrix{
\mathfrak{string}_{\mathrm{IIA}}\ar[rr]\ar[d] &&0\ar[d]\\
\mathbb{R}^{9,1\vert \mathbf{16}+\overline{\mathbf{16}}}\ar[rr]^-{\mu_{{}_{F1}}}
&& b^2 \R\;.
}
\]
\end{defn}

\begin{remark}
By \cite[Prop. 3.5]{Bouq} the
Chevalley-Eilenberg algebra of $\mathfrak{string}_{\mathrm{IIA}}$ is obtained by adjoining  to ${\rm CE}(\mathbb{R}^{9,1\vert \mathbf{16} + \overline{\mathbf{16}}})$ an element $f_2$ of degree $(2,\mbox{even})$ whose CE-differential is the 3-cocycle
$-\mu_{{}_{F1}}$:

\begin{displaymath}
{\rm CE} \big(
    \mathfrak{string}_{\mathrm{IIA}}
  \big)
  =
  \big(
    {\rm CE}(\mathbb{R}^{9,1\vert \mathbf{16}+\overline{\mathbf{16}}})
    \otimes
   \mathbb{R}[ f_2 ]
    \,;~
    d f_2 = -\mu_{{}_{F1}}
  \big)
  \,.
\end{displaymath}
\end{remark}

\medskip
Here the generator $f_2$ will play the role of the field strength of the Chan-Paton gauge field on the
D-branes. We need to add the corresponding fields strengths of the RR-forms. To see which form they
should have, we look at the dimensional reduction of the M-brane charges. By direct computation one finds
\begin{prop}
  \label{reductionOfMBraneCocycles}
  The image of the the M-brane cocycle $(\mu_4,\mu_7): \mathbb{R}^{10,1\vert \mathbf{32}} \longrightarrow \mathfrak{l}(S^4)$
  from Cor. \ref{mu4-mu7IsS4Cocycle} under the double dimensional reduction isomorphism of Prop. \ref{Hom-Prop}
  have the following components:
  $$
    \begin{aligned}
      \mu_{{}_{F1}} & = i \left(\overline{\psi}\wedge \Gamma_{a}\Gamma_{10}\psi\right) \wedge e^a
      \\
      \mu_{{}_{D0}} & = \overline{\psi} \wedge \Gamma_{10} \psi
      \\
      \mu_{{}_{D2}} & = \tfrac{i}{2} \left(\overline{\psi} \wedge \Gamma_{a_1 a_2}\psi\right)
      \wedge e^{a_1} \wedge e^{a_2}
      \\
      \mu_{{}_{D4}} &= \tfrac{1}{4!}\left(\overline{\psi} \wedge \Gamma_{a_1 \cdots a_4}\psi\right) \wedge e^{a_1}\wedge \cdots e^{a_4}\;,
    \end{aligned}
  $$
  where the subscripts reflect their interpretation according to Def. \ref{SelfDuality}.
\end{prop}
In view of Prop. \ref{reductionOfMBraneCocycles} we set:
\begin{defn}
 \label{IIABraneCocycles}
Let $C \in {\rm CE}(\mathfrak{string}_{\mathrm{IIA}})$ be given by
\begin{displaymath}
\begin{aligned}
  C
  \;
  & :=
    \overline{\psi}\Gamma_{10}\psi
   + \tfrac{i}{2}
    \overline{\psi} \Gamma^{a_1 a_2}  \psi\wedge e_{a_1} \wedge e_{a_2}
    +
    \tfrac{1}{4!}
    \overline{\psi} \Gamma^{a_1 ... a_4} \Gamma_{10} \psi \wedge e_{a_1} \wedge \cdots \wedge e_{a_4}
   \\
    & +
    \tfrac{i}{6!}
    \overline{\psi} \Gamma^{a_1 ... a_6}  \psi \wedge e_{a_1} \wedge \cdots \wedge e_{a_6}
    +
    \tfrac{1}{8!}
    \overline{\psi} \Gamma^{a_1 ... a_8} \Gamma_{10}  \psi \wedge e_{a_1} \wedge \cdots \wedge e_{a_8}\;.
  \end{aligned}
\end{displaymath}
\end{defn}
Using the above element we will now introduce a cocycle associated naturally with type IIA
D-branes in a background B-field.
\begin{defn}
\label{IIADpbranecocycles}
For $p \in \{0,2,4,6,8\}$ we  define the $Dp$-brane cocycle
$\mu_{{}_{Dp}} \in {\rm CE}(\mathfrak{string}_{\mathrm{IIA}})$
to be given by
$$
\mu_{{}_{Dp}}
  :=
  \left[
     C \wedge \exp(f_2)
  \right]_{p+2}\;,
$$
where
$$
\exp(f_2)
  :=
  1 + f_2 + \tfrac{1}{2} f_2 \wedge f_2 + \tfrac{1}{6} f_2 \wedge f_2 \wedge f_2 + \cdots \;,
$$
and where the square brackets indicate picking out the homogeneous summand of degree $p+2$.
\end{defn}

 \begin{remark}\label{closed}
The elements $\mu_{{}_{Dp}}$ in Definition \ref{IIADpbranecocycles} are non-trivial cocycles, i.e., they are closed and non-exact elements in ${\rm CE}(\mathfrak{string}_{\mathrm{IIA}})$. This is proved in \cite[Sec. 6.1]{CAIB99} (Beware Remark \ref{dea} when comparing prefactors.)
The action functionals induced by these cocycles $\mu_{{}_{Dp}}$
as in \cite{Bouq} are the WZW-terms for the Green-Schwarz-type sigma-model of the D-branes in type IIA super Minkowski spacetime. The element
$f_2$ represents the field strength of the Chan-Paton gauge field on the D-brane and $C$ is the contribution  of the
Ramond-Ramond fields.
\end{remark}

\begin{remark}
Notice how, in the notation of Proposition \ref{Hom-Prop}, we have, up to prefactors,
\[
C = \omega_2 + \omega_4 + \omega_6 +  \text{higher order terms}\;,
\]
so that
\[
\mu_{{}_{D_0}}= \omega_2;
 \qquad \mu_{{}_{D_2}}= \omega_4 + f_2\wedge \omega_2;
 \qquad \mu_{{}_{D_4}}= \omega_6 + f_2\wedge \omega_4 +\tfrac{1}{2}f_2\wedge f_2\wedge \omega_2\;.
\]
Consequently,  the condition that $\mu_{D_p}$ is a cocycle in $\mathfrak{string}_{\mathrm{IIA}}$ translates to the equations
\[
d\omega_2 = 0\;,
 \qquad d\omega_4 = \mu_{{}_{F_1}}\wedge \omega_2\;,
\qquad d\omega_6 = \mu_{{}_{F_1}} \wedge \omega_4
\]
in ${\rm CE}(\mathbb{R}^{9,1\vert \mathbf{16}+\overline{\mathbf{16}}})$. Indeed, these are precisely
 the equations obtained in Proposition \ref{Hom-Prop}.
\end{remark}

\medskip
We will now start making explicit the connection to (twisted) K-theory at the rational level.

\begin{defn} Define $\mathfrak{l}(\mathrm{ku})$ to be the $L_\infty$-algebra
\[
  \mathfrak{l}(\mathrm{ku})=\bigoplus_{p\text{ even}}b^{p+1}\R\;.
\]
\end{defn}

\begin{remark}
 {\bf (i)} The $L_\infty$-algebra $\mathfrak{l}(\mathrm{ku})$ is a minimal Sullivan model for the  rationalization of the
 connective K-theory spectrum.
\item {\bf (ii)}
Notice that the Chevalley-Eilenberg algebra of $\mathfrak{l}(\mathrm{ku})$ is
\[
{\rm CE}(\mathfrak{l}(\mathrm{ku}))= \big(\R[ \{\omega_{2p}\}_{p = 1,2,\dots}]\,;\, d\omega_{2p}=0\big)\;,
\]
i.e., the even closed forms, as appropriate for rationalization of K-theory, via the Chern character, with
target even rational cohomology.
\item {\bf (iii)}
The direct sum of cocycles
\[
\mu_{{}_D}=\underset{p=0,2,4,6,8}{\bigoplus} \mu_{{}_{D p}}
\]
defines an $L_\infty$-morphism $\mu_{{}_D}\colon \mathbb{R}^{9,1\vert \mathbf{16} + \overline{\mathbf{16}}}\to
\bigoplus_{p=0, 2, 4, 6, 8}b^{p+1}\R \hookrightarrow  \mathfrak{l}(\mathrm{ku})$. Hence we see that
super Minkowski spacetime $\mathbb{R}^{9,1\vert \mathbf{16} + \overline{\mathbf{16}}}$, locally
modeling super spacetimes in 10d type IIA supergravity, fits into a diagram of super
$L_\infty$-morphisms of the form

\begin{displaymath}
\xymatrix{
  \mathfrak{string}_{\mathrm{IIA}} \ar[dd]_-{{\rm hofib}(\mu_{{}_{F1}})}\ar[rrd] \ar[rrrr]^-{\mu_{{}_{D}}}
  &&&& \mathfrak{l}(\mathrm{ku}) \ar[lld]
\\
 & &0\ar[dd]
 \\
 \mathbb{R}^{9,1\vert \mathbf{16}+ \overline{\mathbf{16}}}  \ar[drr]_{\mu_{{}_{F1}}}
 &&&&  \\
 &&
   b^2 \mathbb{R}  \;,&&
}
\end{displaymath}
where the left bottom square is a homotopy pullback.
\end{remark}

By the above remark, it is therefore natural to ask whether the the cocycles $\mu_{{}_{D p}}$ for the D-branes are $b \mathbb{R}$-equivariant and so descend to super-Minkowski spacetime as twisted cocycles, in the
sense of \cite{NSS12}, and in analogy to the descent of the M5-brane cocycle considered in section \ref{Sec M}.
 More explicitly, we are asking whether we can complete the above diagram to a commutative diagram of the form
\begin{displaymath}
\xymatrix{
  \mathfrak{string}_{\mathrm{IIA}}
     \ar[dd]_-{{\rm hofib}(\mu_{{}_{F1}})}\ar[rrd] \ar[rrrr]^-{\mu_{{}_{D}}}
  &&&&
  \mathfrak{l}(\mathrm{ku})
    \ar[lld]\ar[dd]^-{{\rm hofib}(\phi)}
\\
 & &0\ar[dd]
 \\
 \mathbb{R}^{9,1\vert \mathbf{16}+ \overline{\mathbf{16}}}
   \ar[drr]_{\mu_{{}_{F1}}}
   \ar[rrrr]^<<<<<<<<<<<{\mu_{{}_{F1/D}}^{\mathrm{IIA}}}|(.46)\hole|(.47)\hole|(.49)\hole|(.50)\hole|(.51)\hole
 &&&& \mathfrak{l}(\mathrm{ku}/\!/BU(1))\;,\ar[dll]^{\phi} \\
&&
  b^2 \R &&
}
\end{displaymath}
for a suitable $L_\infty$-algebra $\mathfrak{l}(\mathrm{ku}/\!/BU(1))$, in such a way that
both front faces of the prism are homotopy pullbacks.

\begin{defn}
\label{AbstractTwistedDeRham}\hypertarget{AbstractTwistedDeRham}{}
Write $\mathfrak{l}(\mathrm{ku}/\!/BU(1))$
for the $L_\infty$-algebra whose Chevalley-Eilenberg algebra has generators
$\omega_{2p}$ for $p$ a positive integer, and $h_3$, each in the degree indicated by its subscript,
with non-trivial differential given by $d(\omega_{2(p+1)}) = h_3 \wedge \omega_{2p}$:

\begin{displaymath}
{\rm CE}
  \big(
    \mathfrak{l}(\mathrm{ku}/\!/BU(1))    \big)
  :=
  \left\{
    \R[ \{\omega_{2p}, h_3\}_{p = 1,2,\dots}] \;;\, dh_3=0, ~ d \omega_{2(p+1)} = h_3 \wedge \omega_{2p}
  \right\}
  \,.
\end{displaymath}
\end{defn}

\medskip
The following proposition shows that the three $L_\infty$-algebras are compatible in a nice way,
in the sense that indeed $\mathfrak{l}(\mathrm{ku}/\!/BU(1))$ is a rational quotient of $\mathfrak{l}(\mathrm{ku})$
by $\mathfrak{l}(BU(1))$.

\begin{prop}
\label{}\hypertarget{}{}
There is a natural homotopy fiber sequence of $L_\infty$-algebras
\[
\xymatrix{
  \mathfrak{l}(\mathrm{ku})
    \ar[r]\ar[d]&
    0\ar[d]\\
  \mathfrak{l}(\mathrm{ku}/\!/BU(1))\ar[r]
   &
   b^2 \R\;.
}
\]
\end{prop}
\begin{proof}
Passing to the dual Chevalley-Eilenberg algebras, we have a natural commutative diagram
 \[
\xymatrix{
\big( \mathbb{R}[ h_{3}];~ dh_{3}=0\big)\ar[r]\ar[d]&  \big(\R[\{\omega_{2p}, h_3\}_{p = 1,2,\dots}]\;;\, dh_3=0,\, d \omega_{2(p+1)} = h_3 \wedge \omega_{2p}\big)\ar[d]\\
\R\ar[r]& \big(\R[ \{\omega_{2p}\}_{p = 1,2,\dots}]\,;\, d\omega_{2p}=0\big)\;,
}
\]
where each morphism maps a generator to the generator with the same name, if present, and to zero otherwise. This diagram is clearly a pushout. Moreover, since the top horizontal morphism is a relative Sullivan algebra, it is also a homotopy pushout. Therefore, its dual diagram is a homotopy pullback.
\end{proof}
%
%
%
%

Now we connect back to super Minkowski spacetime. To exhibit the morphism
$\mathbb{R}^{9,1\vert \mathbf{16}+ \overline{\mathbf{16}}}\to \mathfrak{l}(\mathrm{ku}/\!/BU(1))$
in the homotopy category of $L_\infty$-algebras, we consider a resolution $\mathbb{R}^{9,1\vert \mathbf{16} + \overline{\mathbf{16}}}_{\mathrm{res}}$ of $\mathbb{R}^{9,1\vert \mathbf{16}+ \overline{\mathbf{16}}}$:
\begin{defn}
\label{ResolutionOfTypeIIASuperLieAlgebra}\hypertarget{ResolutionOfTypeIIASuperLieAlgebra}{}
Write $\mathbb{R}^{9,1\vert \mathbf{16} + \overline{\mathbf{16}}}_{\mathrm{res}}$ for the super
$L_\infty$-algebra whose Chevalley-Eilenberg algebra is obtained by adding to the Chevalley-Eilenberg
algebra of $\mathbb{R}^{9,1\vert \mathbf{16} + \overline{\mathbf{16}}}$ two
generators $f_2$ and $h_3$, of degrees $2$ and $3$ respectively, with $d f_2 = h_3- \mu_{{}_{F1}}$
and $dh_3=0$:
\[
\mathrm{CE}(\mathbb{R}^{9,1\vert \mathbf{16} + \overline{\mathbf{16}}}_{\mathrm{res}}):=
\Big({\rm CE}(\mathbb{R}^{9,1\vert \mathbf{16} + \overline{\mathbf{16}}})\otimes\R[f_2, h_3]\;;
~
d f_2 = h_3- \mu_{{}_{F1}}, ~ dh_3=0
  \Big)\;.
\]
\end{defn}

\begin{remark}\label{differential}
The DG-algebra $\mathrm{CE}(\mathbb{R}^{9,1\vert \mathbf{16} + \overline{\mathbf{16}}}_{\mathrm{res}})$ is isomorphic to the graded commutative algebra
\[\mathrm{CE}(\mathfrak{string}_{\mathrm{IIA}})\otimes \R[h_3]
\] endowed with the differential
\[
d=d_{\mathrm{IIA}} + h_3 {\tiny \frac{\partial}{\partial f_2}}\;,
\]
where $d_{\mathrm{IIA}}$ is the differential of $\mathrm{CE}(\mathfrak{string}_{\mathrm{IIA}})$.
\end{remark}

\medskip
In order to complete our diagrams, we will follow similar steps to the ones we
established for the M-theory supercocycles in the previous section.

\begin{prop}
The canonical morphism $\mathbb{R}^{9,1\vert \mathbf{16} + \overline{\mathbf{16}}}_{\mathrm{res}} \overset{\simeq}{\longrightarrow}
\mathbb{R}^{9,1\vert \mathbf{16} + \overline{\mathbf{16}}}$, dual to the obvious inclusion  ${\rm CE}(\mathbb{R}^{9,1\vert \mathbf{16} + \overline{\mathbf{16}}})\hookrightarrow {\rm CE}(\mathbb{R}^{9,1\vert \mathbf{16} + \overline{\mathbf{16}}}_{\mathrm{res}})$, is
an equivalence of $L_\infty$-algebras. It factors the natural projection
$\mathfrak{string}_{\mathrm{IIA}} \to
\mathbb{R}^{9,1\vert \mathbf{16} + \overline{\mathbf{16}}}$ through the morphism
$\mathfrak{string}_{\mathrm{IIA}} \to \mathbb{R}^{9,1\vert \mathbf{16} + \overline{\mathbf{16}}}_{\mathrm{res}}$ whose dual map is
\begin{align*}
{\rm CE}(\mathbb{R}^{9,1\vert \mathbf{16} + \overline{\mathbf{16}}}_{\mathrm{res}})&
\longrightarrow {\rm CE}(\mathfrak{string}_{\mathrm{IIA}})\\
f_2 &\longmapsto f_2\\
h_3 &\longmapsto 0
\end{align*}
and the identity on all other generators.
\end{prop}
\begin{proof}Use the same strategy as in the proof of Proposition \ref{resolution-one}.
\end{proof}

\begin{prop}
The 3-cocycle $\tilde{\mu}_{{}_{F1}}\colon \mathbb{R}^{9,1\vert \mathbf{16} + \overline{\mathbf{16}}}_{\mathrm{res}}\to b^2 \R$
dual to the obvious inclusion $\R[h_3]\to {\rm CE}(\mathbb{R}^{9,1\vert \mathbf{16} + \overline{\mathbf{16}}}_{\mathrm{res}})$ fits into a homotopy commutative diagram of super $L_\infty$-algebras
\[
\xymatrix{ \mathbb{R}^{9,1\vert \mathbf{16} + \overline{\mathbf{16}}}_{\mathrm{res}}\ar[r]^{\simeq}\ar[dr]_{\tilde{\mu}_{{}_{F1}}}&  \mathbb{R}^{9,1\vert \mathbf{16} + \overline{\mathbf{16}}}
\ar[d]^{\mu_{{}_{F1}}}\\
& b^2 \R\;.
}
\]
\end{prop}
\begin{proof}
Use the same strategy as in the proof of Proposition \ref{equivalent-cocycle}.
\end{proof}

With the above results, we are now able to complete the desired diagram:
\begin{theorem}
\label{descentoftheIIAcocycles}\hypertarget{}{}
The F1/Dp-brane cocycles of type IIA fit into a commutative diagram of super $L_\infty$-algebras
\begin{displaymath}
\xymatrix{
&  \mathfrak{string}_{\mathrm{IIA}} \ar[dd]_-{{\rm hofib}(\mu_{{}_{F1}})}\ar[rrd] \ar[rrrr]^-{\mu_{{}_{D}}} &&&& \mathfrak{l}(\mathrm{ku}) \ar[lld]\ar[dd]^-{{\rm hofib}(\phi)}
\\
& & &0\ar[dd]
 \\
 \mathbb{R}^{9,1\vert \mathbf{16}+ \overline{\mathbf{16}}}& \mathbb{R}^{9,1\vert \mathbf{16}+ \overline{\mathbf{16}}}_{\mathrm{res}} \ar[l]_{\simeq} \ar[drr]_{\tilde{\mu}_{{}_{F1}}}
 \ar[rrrr]^<<<<<<<<<<<{\mu_{{}_{F1/D}}^{\mathrm{IIA}}}|(.46)\hole|(.47)\hole|(.49)\hole|(.50)\hole|(.51)\hole^{}
 &&&& \mathfrak{l}(\mathrm{ku}/\!/BU(1))\;,\ar[dll]^{\phi} \\
&&& b^2 \R&&
}
\end{displaymath}
where the two front faces of the prism are homotopy pullbacks.
\end{theorem}

\begin{proof}
The dual diagram of the prism is the following diagram of CE-algebras

\begin{displaymath}
\xymatrix{
  \Big({
    {\rm CE}(\mathbb{R}^{9,1\vert \mathbf{16}+\overline{\mathbf{16}}})
    \otimes
   \mathbb{R}[ f_2 ]
    \,;
    \atop{d f_2 = -\mu_{{}_{F1}}}}
  \Big)
  &&&& \Big({\R[ \{\omega_{2p}\}_{p = 1,2,\dots}]\,;\,\atop{ d\omega_{2p}=0}}\Big)
  \ar[llll]_-{_{( \omega_{2(p+1)} \mapsto \mu_{{}_{D p}})_{p=0,1,2,3,4}}}
\\
 & &\R\ar[rru]\ar[llu]
 \\
 \left({{\rm CE}(\mathbb{R}^{9,1\vert \mathbf{16} + \overline{\mathbf{16}}})\otimes\R[f_2, h_3]\;;
\, \atop{
d f_2 = h_3- \mu_{{}_{F1}};\,\, dh_3=0}}\right)
    \ar[uu]^-{h_3\mapsto 0}
 &&&&  \left({
    \R[ \{\omega_{2p}, h_3\}_{p = 1,2,\dots}] \;;\atop{\, dh_3=0;\, d \omega_{2(p+1)} = h_3 \wedge \omega_{2p}}}
  \right)\ar[llll]|(.49)\hole|(.50)\hole|(.51)\hole|(.52)\hole_{( \omega_{2(p+1)} \mapsto
   \mu_{{}_{D p}})_{p=0,1,2,3,4}}\ar[uu]_-{h_3\mapsto0}\;. \\
&& \big( \mathbb{R}[ h_{3}]\,;~dh_{3}=0\big)\ar[ull]^{}\ar[urr]\ar[uu]|(.57)\hole&&
}
\end{displaymath}

\noindent One only needs to check that the bottom horizontal map is indeed a homomorphism of DGCAs. To see this, recall from Remark \ref{differential} that the   differential of $\mathrm{CE}(\mathbb{R}^{9,1\vert \mathbf{16} + \overline{\mathbf{16}}}_{\mathrm{res}})$ may be written as

\begin{displaymath}
d = d_{\mathrm{IIA}} + h_3 \wedge \frac{\partial}{\partial f_2}
  \,,
\end{displaymath}
and from Remark \ref{closed} that the elements $\mu_{{}_{Dp}}$ are closed in ${\rm CE}(\mathfrak{string}_{\mathrm{IIA}})$. Then, recalling the definition of the Dp-brane cocycles $\mu_{{}_{Dp}}$, i.e.,
$
\mu_{{}_{Dp}} := [C \wedge \exp(f_2)]_{p+2}
$,
we see that in ${\rm CE}(\mathbb{R}^{9,1\vert \mathbf{16}+\overline{\mathbf{16}}}_{\mathrm{res}})$ we have
\begin{displaymath}
\begin{aligned}
    d\mu_{{}_{D2(p+1)}}
    & =
    \underset{= 0}{\underbrace{d_{\mathrm{IIA}}(\mu_{{}_{D2(p+1)}})}}
    +
    h_3 \wedge \frac{\partial}{\partial f_2} [C \wedge \exp(f_2)]_{2p+4}
    \\
    & =
    [C \wedge h_3 \wedge \exp(f_2)]_{2p+5}
    \\
    & = h_3 \wedge [C \wedge \exp(f_2)]_{2p+2}
    \\
    & = h_3 \wedge \mu_{{}_{D2p}}\;.
  \end{aligned}
\end{displaymath}
This equation precisely says that the bottom horizontal arrow preserves the differentials.
\end{proof}


\medskip
The above result demonstrates that the type IIA F1-brane and D-brane cocycles with
$\mathbb{R}$-coefficients indeed descend to super-Minkowski spacetime as one single
cocycle with coefficients in the homotopy quotient
$\Big(
       \underset{p = 0,2,4,6,8}{\bigoplus}
       \mathbf{B}^{2p+1}\mathbb{R}
     \Big)
     /\!/
     \mathbf{B} \mathbb{R}$.

\medskip

We close by making explicit how this is a rational model for twisted K-theory.
For this we use the general theory of twisted cohomology from \cite{NSS12}:
Let $A$ be any homotopy type which represents some cohomology theory
(for instance the degree-0 space in a spectrum).
Then a higher homotopy action $\rho$ of some $\infty$-group
$G$ (for instance the homotopy type of some topological group)
on $A$ is equivalently encoded in a homotopy fiber sequence of the form
$$
  \xymatrix{
    A \ar[r] & A/\!/G \ar[d]^{p_\rho}
    \\
    & ~B G\;,
  }
$$
where the homotopy type $A/\!/G$ is identified thereby with the homotopy quotient of $A$ by the
$\infty$-action of $G$ \cite[section 4.1]{NSS12}.

\medskip
We may alternatively think of this as exhibiting an $A$-fiber $\infty$-bundle $A/\!/G$ over $B G$,
namely the $A$-fiber bundle which is associated via the given action to the universal
principal $G$-$\infty$-bundle $E G$ (the Borel construction):
$$
  A/\!/G \simeq E G \times_G A\;.
$$
Given this data, then a twist for the $A$-cohomology of a space $X$ is equivalently a map
$\tau \colon X \longrightarrow B G$, hence a $G$-principal $\infty$-bundle $P_\tau \to X$,
and a cocycle in $\tau$-twisted $A$-cohomology of $X$ is a diagram as on the left in the following
$$
  \left\{
  \raisebox{20pt}{
  \xymatrix{
    X \ar[dr]_{\tau}^{\ }="t" \ar@{-->}[rr]^c_{\ }="s" && A/\!/G \ar[dl]^{p_\rho}
    \\
    & B G
    %
  }
  }
  \right\}
   \;\;\;\;
   \simeq
   \;\;\;\;
   \left\{
   \raisebox{20pt}{
   \xymatrix{
     & P_\tau \times_G A \ar[d]
     \\
     X \ar@{-->}[ur]^-{\sigma} \ar@{=}[r] & X
   }
   }
   \right\}\;.
$$
This is discussed in \cite[section 4.2]{NSS12}. Equivalently, as shown on the right,
this is a section $\sigma$ of the $A$-fiber $\infty$-bundle
$P_\tau \times_G A$
that is $\rho$-associated to the twist bundle $P_\tau$.
\footnote{Here and in the following all diagrams are filled by specified homotopies, which we
notationally suppress.}
This second perspective of twisted cohomology, in terms of sections of bundles
of coefficients spaces (i.e., ``local coefficients'') is prominently
reflected in the literature on twisted ordinary cohomology and twisted
K-theory. For our purposes here the equivalent perspective on the left above, in terms of
maps into homotopy quotients is more directly compared to the
rational data which we derived above.

\medskip
Now consider the case that $A = \Omega^\infty \mathcal{A}$ is the degree-0
space of a spectrum $\mathcal{A}$. Then if the action of $G$ extends to an
action on the spectrum, then $\mathcal{A}/\!/G$ becomes a \emph{parameterized spectrum}
over $B G$, namely a sequence of spaces $(\mathcal{A}/\!/G)_n$ for $n \in \mathbb{N}$,
each equipped with a retraction onto $B G$
$$
  \mathrm{id} : B G \longrightarrow (\mathcal{A}/\!/G)_n \longrightarrow B G
$$
and equipped with weak equivalences
$$
  (\mathcal{A}/\!/G)_n
    \stackrel{\simeq}{\longrightarrow}
  B G \underset{(\mathcal{A}/\!/G)_{n+1}}{\times} B G
$$
into the homotopy fiber product of the map $B G \to (\mathcal{A}/\!/G)_{n+1}$ with itself
(see \cite[section 3.4]{AndoBlumbergGepner} \cite[section 4.1.2]{dcct}).

\medskip
In the special case that $G \simeq \ast$ and hence $B G \simeq \ast$ then this reduces to an
ordinary Omega-spectrum, thought of as parameterized over the point.
On the other extreme, the space $B G$ we may think of as the 0-spectrum parameterized
over $B G$, by taking all the structure morphisms above to be equivalences.

\medskip
This way the parameterized spectrum $\mathcal{A}/\!/G$ sits in a homotopy fiber sequence
of the form
$$
  \raisebox{20pt}{
  \xymatrix{
    \mathcal{A}\ar[r] & \mathcal{A}/\!/G \ar[d]
    \\
    & B G
  }
  }
$$
in direct analogy with the unstable situation above. This now classifies twisted \emph{stable}
cohomology \cite[sect 4.1.2.1]{dcct}.

\medskip
Now suppose that $\mathcal{A}$ is a commutative ring spectrum. Then there
is an $\infty$-group
$$
  \mathrm{GL}_1(\mathcal{A}) := \Omega^\infty \mathcal{A} \times_{\pi_0(\mathcal{A})} \pi_0(\mathcal{A})^\times
$$
(the ``$\infty$-group of units'' \cite[p. 10]{AndoBlumbergGepner})
acting on $\mathcal{A}$ by the homotopy theoretic analog of the action of the group of units
of a commutative ring.
For the case that $\mathcal{A} = \mathrm{KU}$ is complex K-theory,
 $B GL_1(\mathrm{KU})$ receives a non-trivial map from $B^2 U(1)$
classifying the twist of K-theory by classes in degree-3 integral cohomology (\cite[p. 15]{AndoBlumbergGepner}).
Restricting the
$GL_1(\mathrm{KU})$-action along this inclusion hence exhibits an $\infty$-action of $B U(1)$
on complex K-theory
$$
  \xymatrix{
    \mathrm{K U}
    \ar[r]
    &
    \mathrm{K U}/\!/B U(1)
    \ar[d]
    \\
    & B^2 U(1)\;.
  }
$$
This means that maps $\tau : X \to B^2 U(1)$ (classifying $U(1)$-bundle gerbes) are twists for $K$-theory,
and that fixing one such twist $\tau$ then cocycles in $\tau$-twisted K-theory are
diagrams as on the left of the following:
$$
  \left\{
  \raisebox{20pt}{
  \xymatrix{
    X \ar[dr]_\tau^{\ }="t" \ar@{-->}[rr]^-c_{\ }="s" && \mathrm{KU}/\!/BU(1) \ar[dl]^{p_\rho}
    \\
    & B^2 U(1)
    %
  }
  }
  \right\}
  \;\;\;
  \simeq
  \;\;\;
  \left\{
  \raisebox{20pt}{
  \xymatrix{
    & P_\tau \underset{B U(1)}{\times} \mathrm{K U}
     \ar[d]
    \\
    X
    \ar@{=}[r]
    \ar@{-->}[ur]^-\sigma
    &
    X
  }
  }
  \right\}
  \,.
$$
On the right we see how this is equivalently given by sections of the
$\mathrm{KU}$-fiber bundle that is associated to the twist bundle $P_\tau$.
This is the perspective on twisted K-theory from \cite[section 3.4]{AndoBlumbergGepner}.
Here we need the equivalent perspective on the left.
The equivalence between the two perspectives is \cite[proposition 4.17]{NSS12}.

\medskip
Now we may compare the above diagrams for twisted theory  with the result for the descended $L_\infty$-cocycles of the type IIA
D-branes from Theorem \ref{descentoftheIIAcocycles}.

\medskip
\medskip
\begin{tabular}{cc}
 \fbox{
 $ \raisebox{20pt}{
  \xymatrix{
    && \mathrm{KU}
     \ar[d]^{\mathrm{hofib}(p_\rho)}
    \\
    X \ar[dr]_\tau^{\ }="t" \ar@{-->}[rr]^c_{\ }="s" && \mathrm{KU}/\!/BU(1) \ar[dl]^{p_\rho}
    \\
    & B^2 U(1)
    %
  }
  }
  $
  }
  &
  \fbox{
  $
  \raisebox{20pt}{
  \xymatrix{
    && \mathfrak{l}(\mathrm{ku})
    \ar[d]^{\mathrm{hofib}(\phi)}
    \\
    \mathbb{R}^{9,1\vert \mathbf{16}+ \overline{\mathbf{16}}}
    \ar[rr]^{\mu_{{}_{F1/D}}^{\mathrm{IIA}}}
    \ar[dr]_{\mu_{{}_{F1}}}
    &&
    \mathfrak{l}( \mathrm{ku}/\!/ BU(1) )
    \ar[dl]^{\phi}
    \\
    & b^2 \mathbb{R}
  }
  }
  $
  }
 \\
 \\
 \begin{tabular}{c}
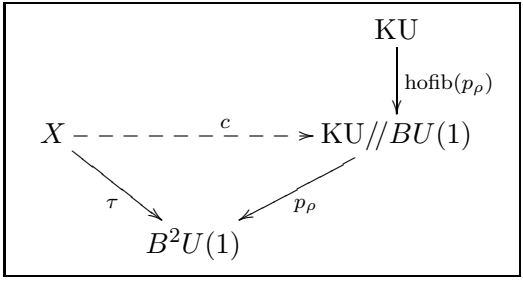
   {\it A cocycle in twisted K-theory}
   \\
{\it   according to \cite{AndoBlumbergGepner, NSS12}.}
 \end{tabular}
 &
 \begin{tabular}{c}
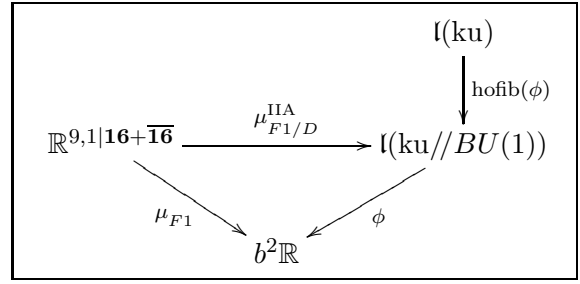
  {\it The descended type IIA F1/D$p$-brane cocycle}
  \\
{\it  according to Theorem \ref{descentoftheIIAcocycles}.}
 \end{tabular}
\end{tabular}

\medskip
\medskip
It is clear that $\mathrm{CE}(\mathfrak{l}(\mathrm{ku}))$ is the Sullivan model for
$(\Omega^\infty \mathrm{KU})_{\mathbb{R}}$. Hence this shows that the descended super
$L_\infty$-cocycles for F1/D$p$-branes in type IIA take values in the rationalization
of the classifying space for twisted K-theory.

\appendix

\section{Basics of Rational Homotopy Theory}

Much of the difficulty in homotopy theory is determined by torsion phenomena (i.e. the appearance of additively nilpotent elements
in homotopy groups and cohomology groups).
And actually, torsion is really mysterious. The following a priori very surprising statement
 parametrizes the situation: there is no non-contractible simply connected manifold $X$ for which
the torsion part of the abelian homotopy groups $\pi_n(X)$ is known for all $n$. In other words, the knowledge of  torsion in homotopy is always approximate.
In contrast, the study of the free part of the homotopy groups can be done
in a systematic way \cite{Q, Su}.
One way of doing so is to notice that, if we write
\[
\pi_n(X)=\pi_n(X)_{\mathrm{free}}\oplus \pi_n(X)_{\mathrm{tors}}=\mathbb{Z}^{\mathrm{rank(\pi_n(X))}}\oplus \pi_n(X)_{\mathrm{tors}},
\]
then by tensoring with the field $\Q$ of rational numbers we kill the torsion part and get
$$
\pi_n(X) \otimes \Q=\Q^{\mathrm{rank(\pi_n(X))}}.
$$
Wouldn't it be nice if the right hand side were the actual homotopy group of a
space $X_{\Q}$ which would then be a rational stand-in for the space $X$? Indeed, every simply-connected space admits
such a \emph{rationalization} or \emph{$\Q$-localization} $X_{\Q}$, and this can be constructed in a functorial way.
Two spaces $X$ and $Y$ are then said to be rationally homotopy equivalent,
$X \sim_\Q Y$, if their rationalizations $X_\Q$ and $Y_\Q$ are homotopy
equivalent. Rational homotopy theory is the study of spaces up to rational
homotopy equivalence, i.e., informally, after stripping off torsion. All background on this can
be found in the excellent treatments in
\cite{BG, FHT, FHT2, FOT, GM, He, HMR}. Here we recall a few basics which may help the non-expert reader through the main body of this article.

\medskip
To a simply-connected topological space $X$ one may associate an
algebraic object, called a minimal model for $X$, such that
$X\sim_\Q Y$ if and only if the minimal models for $X$ and $Y$ are
isomorphic.
The reason one restricts to simply connected spaces is that one needs the fundamental group $\pi_1(X)$ to act trivially on all the homotopy groups $\pi_n(X)$. This condition (which in particular implies that $\pi_1(X)$ needs to be abelian) is trivially satisfied by simply connected spaces; however there are remarkable examples of non-simply connected spaces satisfying it, notably non-simply connected topological groups. In what follows in this Appendix we will keep assuming that the spaces we will be dealing with are simply connected, but the reader should keep in mind that the results hold more generally for spaces with a trivial $\pi_1$-action.

\medskip
There are two main approaches to algebraic minimal models for topological spaces, the first
cohomological and the second homological:
\begin{enumerate}
\item Sullivan \cite{Su} uses  differential graded commutative algebras (DGCAs),
\item Quillen \cite{Q} uses differential graded Lie algebras (or more generally $L_\infty$-algebras).
\end{enumerate}
The two approaches are dual to each other, via the Koszul duality between  differential graded commutative algebras and differential graded Lie algebras. In particular, we will denote by $\mathfrak{l}(X)$ the differential graded Lie algebra (or more generally $L_\infty$-algebra) minimal model (or \emph{Quillen model}) for the simply connected space $X$, so that the  differential graded commutative algebra minimal model (or \emph{Sullivan model}) $A_X$ for $X$ will be $\mathrm{CE}(\mathfrak{l}(X))$, where $\mathrm{CE}(\mathfrak{g})$ denotes the Chevalley-Eilenberg algebra of an $L_\infty$-algebra $\mathfrak{g}$.

\medskip
Differential graded commutative algebras arising as Sullivan minimal models of simply connected spaces are peculiar among all DGCAs. Namely, forgetting the differential, they are
free as graded commutative algebras. In other words, forgetting the differential, Sullivan minimal models are graded polynomial algebras. This means that there exists a graded vector space $V_X$ such that $A_X=\wedge^\bullet V_X$ as graded commutative algebra. The graded vector space $V_X$ is the graded vector space of \emph{generators} for $A_X$. Differential graded commutative algebras of this kind are usually called \emph{semi-free}
or \emph{quasi-free} DGCAs in the mathematical literature. Sullivan minimal models of simply connected spaces have another couple of special features: there are no degree 1 generators and the differential of each generator $x$ is a polynomial in generators $y$ with $\deg(y)<\deg(x)$. An abstract semifree DGCA with these two features is usually called a \emph{Sullivan algebra}.

\medskip
As the association $X\rightsquigarrow A_X$ is contravariant, if $f\colon X\to Y$ is a continuous map, then we have an induced morphism $f^*\colon A_Y\to A_X$ between the Sullivan models. In particular, if $p\colon X\to Y$ is a fibration, then $p^*$ is an inclusion of a special kind: one has $p^*\colon A_Y\hookrightarrow A_Y\otimes_{\Q}\wedge^\bullet V_F\cong  A_X$, where $V_F$ is a graded vector space associated with the fiber $F$ of the fibration $p$. Moreover, the differential in $A_X$ maps a generator $x$ in $V_F$ to a polynomial with coefficients in $A_Y$ in the generators $y$ from $V_F$ with $\deg(y)<\deg(x)$. Abstracting this situation to arbitrary semifree DGCAs one gets the notion of \emph{relative Sullivan algebra}. In the dual picture, these correspond to fibrations between differential graded Lie algebras (or, more gnerally, $L_\infty$-algebras). Notice how a Sullivan algebra $A$ is a relative Sullivan algebra for the inclusion $\mathbb{Q}\hookrightarrow A$. Geometrically, this corresponds to the fibration $X\to *$, where $*$ is the topological space consisting of a single point. In the differential Lie algebra picture, this is the fibration $\mathfrak{g}\to 0$.

\medskip

The differential graded commutative algebra $A_X$ captures all of the homotopy type of $X$ up to torsion. In particular, if $X$ admits a PL-manifold (piece-wise linear) structure, then $A_X$ is quasi-isomorphic to the DGCA  $A_{\mathrm{PL}}(X;\mathbb{Q})$ of piecewise linear differential forms with rational coefficients on $X$. This implies that one has an isomorphism of graded commutative algebras
\[
H^\bullet(A_X)\cong H^\bullet(X;\mathbb{Q}).
\]
Moreover, the subspace $V_X$ of linear generators is canonically identified with the linear dual of rationalized homotopy groups of $X$:
\[
V_X\cong  \bigoplus_{n\in \mathbb{N}} {\rm Hom}_\Z (\pi_n(X), \Q)\;,
\]
provided that $X$ is simply connected and has rational
homology of finite type. This is needed so that  the
corresponding loop space is connected. Notice that this implies that the degree $n$ component of $V_X$ is (noncanonically) isomorphic to $\pi_n(X)\otimes \Q$

\medskip
By extending scalars to $\R$, i.e., by considering the differential graded commutative algebra $A_{X;\R}=A_X\otimes_{\Q}\R=(\wedge^\bullet V_X\otimes_{\Q}\R,d)$, one gets a Sullivan $\mathbb{R}$-algebra with $A_{X;\R}\simeq A_{\mathrm{PL}}(X;\R)$ (and so, in particular, such that $H^\bullet(A_{X;\R})\cong H^\bullet(X;\R)$) and with the degree $n$ component of $V_X\otimes_{\Q}\R$ isomorphic to $\pi_n(X)\otimes_{\mathbb{Z}} \R$. One says that $A_{\R}(X)$ represents the real homotopy type of $X$. When $X$ is a smooth manifold, as are the spaces considered in this paper, the algebra $A_{\mathrm{PL}}(X;\R)$ of piecewise linear forms is quasi-isomorphic to the de Rham algebra $\Omega^\bullet(X;\mathbb{R})$ of dfferential forms, so that $A_{X;\R}$ is a Sullivan model for $\Omega^\bullet(X;\mathbb{R})$ in this case. This is a stronger statement than saying that the cohomology of $A_{X;\R}$ is isomorphic to the de Rham cohomology of $X$.

\medskip
Informally speaking, what the quasi-isomorphism $A_{X;\R}\simeq \Omega^\bullet(X;\mathbb{R})$ gives, is a choice of representative differential forms for the generating cohomology classes of $X$, together with a choice of differential forms representing the algebraic relations between cohomology classes. For physics purposes, this is precisely what one wants: to identify representatives for cohomology classes on the nose, in the spirit and philosophy of  differential cohomology.
For instance, supergravity fields and their dynamics are usually captured by
differential forms and can be promoted to
de Rham cohomology or rational cohomology when taking the gauge
structure into account. Hence it makes sense to
aim to capture this systematically via rational homotopy theory, which is what we do in the present paper.
As we tried to suggest in this short Appendix, the idea is that describing topological aspects of a space is sometimes
easier via algebras, provided we are working rationally.

\medskip
As a matter of notation, since we will always be working over the reals, in the main body of the paper we simply write $A_X$ and $\mathfrak{l}(X)$ for $A_{X;\mathbb{R}}$ and $\mathfrak{l}(X)\otimes_{\mathbb{Q}}\R$, respectively.

\section{Spinors}

Our spinor conventions are as in \cite[II.7.1]{CDF}, except for the first two points in def. \ref{basicconventions} to follow, where we use the opposite signs.
This means that our Clifford matrices behave just as in \cite[II.7.1]{CDF}, the only difference is in a sign when raising a
spacetime index or lowering a spacetime index.
\begin{defn}\label{basicconventions}
{\bf (i)} The Lorentzian spacetime metric is $\eta := \mathrm{diag}  (-1,+1,+1,+1,\cdots)$.
\item
{\bf (ii)} The Clifford algebra relation is $\Gamma_a \Gamma_b + \Gamma_b \Gamma_a  = - 2 \eta_{a b}$;
\item
{\bf (iii)} The timelike index is $a = 0$, the spacelike indices range $a \in \{1,\cdots, d-1\}$.
\item
{\bf (iv)} A unitary Dirac representation of $\mathrm{Spin}(d-1,1)$ is on
$\mathbb{C}^{\nu}$ where $d  \in \{2\nu, 2\nu +1\}$,  via
Clifford matrices such that $\Gamma_0^\dagger = \Gamma_0$ and $\Gamma_a^\dagger = - \Gamma_a^\dagger$
for $a \geq 1$.
\item
{\bf (v)} For $\psi \in \mathrm{Mat}_{\nu \times 1}(\mathbb{C})$ a complex spinor, we write
$\overline{\psi} := \psi^\dagger \Gamma_0$
for its Dirac conjugate. If we have Majorana spinors forming a real sub-representation $S$
then restricted to these the Dirac conjugate coincides with the Majorana conjugate
$\psi^\dagger \Gamma_0 = \psi^T C$ (where $C$ is the Charge conjugation matrix).
\end{defn}
As usual we write
$$
  \Gamma_{a_1 \cdots a_p}
    :=
  \tfrac{1}{p!}
  \sum_{\mbox{\tiny permutations $\sigma$}}
    (-1)^{\vert \sigma\vert}
    \Gamma_{a_{\sigma(1)}} \cdots \Gamma_{a_{\sigma(p)}}
$$
for the anti-symmetrization of products of Clifford matrices.
  These conventions imply that all $\Gamma_a$ are self-conjugate with respect to the pairing $\overline{(-)}(-)$, hence
  that
  $$
    \left(\overline{\psi} \Gamma_{a_1\cdots a_p} \psi\right)^\ast
      =
    (-1)^{p(p-1)/2}\,  \overline{\psi} \Gamma_{a_1 \cdots a_p} \psi
  $$
  holds for all $\psi$.
  This means that the following expressions are real numbers
  $$
     \overline{\psi}\psi
\;, \quad
     \overline{\psi}\Gamma_a \psi
  \;, \quad
     i  \overline{\psi}\Gamma_{a_1 a_2} \psi
  \;, \quad
     i  \overline{\psi}\Gamma_{a_1 a_2 a_3}\psi
   \;, \quad
      \overline{\psi}\Gamma_{a_1\cdots a_4} \psi
    \;, \quad
      \overline{\psi}\Gamma_{a_1\cdots a_5} \psi
    \;, \quad
      \cdots
      \;.
  $$
\begin{defn}\label{SuperMinkowski}
  Given $d\in \mathbb{N}$ and $N$ a real $\mathrm{Spin}(d-1,1)$-representation
  (hence some direct sum of Majorana and Majorana-Weyl representations), the
  corresponding super-Minkowski super Lie algebra
  $$
    \mathbb{R}^{d-1,1\vert N}
      \;\;
      \in
      \mathrm{sLieAlg}_{\mathbb{R}}
  $$
  is the super Lie algebra defined by the fact that its Chevalley-Eilenberg algebra
  is the $(\mathbb{N}, \mathbb{Z}/2)$-bigraded differential-commutative differential algebra
  generated from elements $\{e^a\}_{a = 0}^{d-1}$ in bidegree $(1,\mathrm{even})$
  and from elements $\{\psi^\alpha\}_{\alpha = 1}^{\mathrm{dim}N}$ in bidegree $(1,\mathrm{odd})$
  with differential given by
  $$
    d\psi^\alpha = 0
    \,,
    \qquad
    d e^a = \overline{\psi} \wedge \Gamma^a \psi
    \,.
  $$
  Here on the right we use the spinor-to-vector bilinear pairing, regarded as a super 2-form, i.e.
  in terms of the charge conjugation matrix $C$ this is
  $$
    \overline{\psi} \wedge \Gamma^a \psi
      =
    (C \Gamma^a)_{\alpha \beta} \, \psi^\alpha \wedge \psi^\beta
    \,,
  $$
  where summation over repeated indices is understood.
\end{defn}
\begin{remark}
  \label{dea}
Notice that we omit a factor of $\tfrac{1}{2}$ in the expression for $d e^a$, compared to
the convention in \cite{CDF, CAIB99}.
\end{remark}


\vspace{1cm}
\noindent {\large \bf Acknowledgement}

\medskip

We thank Jonathan Pridham for discussion of fibrations of $L_\infty$-algebras.

\noindent H.S. thanks the Mathematics Department at Rome University La Sapienza,
Insitut Henri Poincar\'e in Paris, and Instituto Superior T\'ecnico, Universidade de Lisboa,
for kind hospitality during the writing of this paper.

\noindent U. S. thanks the Max Planck Institute for Mathematics in Bonn and the Instituto Superior T\'ecnico, Universidade de Lisboa for
kind hospitality during the writing of this paper. U.S. was supported by RVO:67985840.

\bibliography{F1Dp-revised}
\bibliographystyle{plain}

\end{document}